\documentclass{article}

\newcommand{\var}{ \text{\bf var} }
\newcommand{\cov}{ \text{\bf cov} }
\newcommand{\diag}{ \text{\bf diag} }  

\newcommand{\Var}{ \text{\bf Var}  }

\newcommand{\Cov}{ \text{\bf Cov}  }

\usepackage{amsmath}
\usepackage{amsthm}

\DeclareMathOperator*{\argmin}{arg\,min}
\newtheorem{theorem}{Theorem}

\newtheorem{lemma}{Lemma}
\newtheorem{example}[theorem]{Example}
\newtheorem{proposition}[theorem]{Proposition}

\usepackage{enumitem}
\usepackage[ruled,vlined]{algorithm2e}
\usepackage{multirow}
\usepackage{bm}
\usepackage{pdflscape}
\usepackage{graphicx, array}
\usepackage{pgffor}
\usepackage{mathrsfs}
\usepackage{listings}
\foreach \x in {A,...,Z,a,...,z}{%
  \expandafter\xdef\csname cal\x\endcsname{\noexpand %
	\ensuremath{\noexpand\mathcal{\x}}}
  \expandafter\xdef\csname scr\x\endcsname{\noexpand %
	\ensuremath{\noexpand\mathscr{\x}}}
  \expandafter\xdef\csname bb\x\endcsname{\noexpand %
	\ensuremath{\noexpand\mathbb{\x}}}
  \expandafter\xdef\csname rm\x\endcsname{\noexpand %
	\ensuremath{\noexpand\mathrm{\x}}}
  \expandafter\xdef\csname bf\x\endcsname{\noexpand %
	\ensuremath{\noexpand\mbf{\x}}}
}

\newcommand{\bomg}{\bm{\Omega}}
\newcommand{\bsig}{\bm{\Sigma}}
\newcommand{\blam}{\bm{\Lambda}}
\newcommand{\bdel}{\bm{\Delta}}

\newcommand{\bthe}{\bm{\Theta}}

\newcommand{\bphi}{\bm{\Phi}}

\newcommand{\bpi}{\bm{\Pi}}
\newcommand{\bst}{ \hspace{1.5pt} | \hspace{1.5pt} }

\newcommand{\zf}{Z}
\newcommand{\zg}{E}
\newcommand{\zo}{Z_{m}}
\newcommand{\tr}{\text{tr}}

\newcommand{\req}[1]{(\ref{#1})}

\newtheorem{remark}[theorem]{Remark}

\usepackage{makecell}
\usepackage{graphicx}
\usepackage{enumitem}
\usepackage{booktabs}
\usepackage{framed}
\usepackage{multirow}
\usepackage{scrextend}
\usepackage{stackengine}
\usepackage{subcaption}
\usepackage{comment}
\usepackage{url}
\usepackage{bm}
\usepackage[bottom]{footmisc}
\usepackage{amssymb}

\makeatletter
\newcommand{\raisemath}[1]{\mathpalette{\raisem@th{#1}}}
\newcommand{\raisem@th}[3]{\raisebox{#1}{$#2#3$}}
\makeatother

\usepackage{pgffor}

\foreach \x in {A,...,Z,a,...,z}{%
  \expandafter\xdef\csname cal\x\endcsname{\noexpand %
	\ensuremath{\noexpand\mathcal{\x}}}
  \expandafter\xdef\csname scr\x\endcsname{\noexpand %
	\ensuremath{\noexpand\mathscr{\x}}}
  \expandafter\xdef\csname bb\x\endcsname{\noexpand %
	\ensuremath{\noexpand\mathbb{\x}}}
  \expandafter\xdef\csname rm\x\endcsname{\noexpand %
	\ensuremath{\noexpand\mathrm{\x}}}
  \expandafter\xdef\csname bf\x\endcsname{\noexpand %
	\ensuremath{\noexpand\mbf{\x}}}
}

\usepackage{arxiv}

\usepackage[utf8]{inputenc} 
\usepackage[T1]{fontenc}    
\usepackage{hyperref}       
\usepackage{url}            
\usepackage{booktabs}       
\usepackage{amsfonts}       
\usepackage{nicefrac}       
\usepackage{microtype}      
\usepackage{lipsum}
\usepackage{float}

\title{Endogenous Representation of Asset Returns}

\author{
  Zhipu Zhou
  \\
  Department of Statistics \& Applied Probability\\
  University of California, Santa Barbara\\
  Santa Barbara, CA 93106 \\
  \texttt{zzp8933@gmail.com} \\
   \And
Alexander Shkolnik\\
Department of Statistics \& Applied Probability\\
University of California, Santa Barbara\\
Santa Barbara, CA 93106\\
\texttt{shkolnik@ucsb.edu} \\
   \And
 Sang-Yun Oh\\
  Department of Statistics \& Applied Probability\\
  University of California, Santa Barbara\\
  Santa Barbara, CA 93106\\
  \texttt{syoh@ucsb.edu} \\
}

\begin{document}

\maketitle

\begin{abstract}
Factor modeling of asset returns has been a dominant practice in investment science since the introduction of the Capital Asset Pricing Model (CAPM) and the Arbitrage Pricing Theory (APT). The factors, which account for the systematic risk, are either specified or interpreted to be exogenous. They explain a significant portion of the risk in large portfolios. We propose a framework that asks how much of the risk, that we see in equity markets, may be explained by the asset returns themselves. To answer this question, we decompose the asset returns into an endogenous component and the remainder, and analyze the properties of the resulting risk decomposition. Statistical methods to estimate this decomposition from data are provided along with empirical tests. Our results point to the possibility that most of the risk in equity markets may be explained by a sparse network of interacting assets (or their issuing firms). This sparse network can give the appearance of a set exogenous factors where, in fact, there may be none. We illustrate our results with several case studies.
\end{abstract}


\section{Introduction}

It has been a central issue in finance on how to characterize the returns of assets. Harry Markowitz introduced modern portfolio theory in 1952 by formalizing a portfolio of assets as a tradeoff between mean and variance of returns \cite{Markowitz1952}. Building on the sufficiency of the mean-variance framework for investment decision making, \cite{Sharpe1964} and others introduced the theoretical breakthroughs as the Capital Asset Pricing Model (CAPM). CAPM leads to a one-factor model that characterizes the excess return of an asset as the market return in proportion to the beta plus a specific return. In CAPM, the market beta defines the primary force behind equity markets, but what is the market beta and how to identify the market portfolio remain elusive \cite{Roll1977}. As an alternative, the Arbitrage Pricing Theory (APT) proposed by \cite{Ross1976} is based on the asymptotic arbitrage argument rather than the mean-variance optimization framework \cite{Huberman1982}, \cite{Ingersoll1984}. It allows for multiple risk factors in a factor model and does not require to identify the market portfolio. However, how many factors should be included and what are these factors are not answered by APT.

In practice, various statistical tools have been used to estimate the market beta, together with other factors in a factor model, including the least-squares method \cite{Bretscher1995}, maximum likelihood \cite{Rossi2018}, principal component analysis (PCA) \cite{Pearson1901} and so on (see \cite{Chamberlain1983} and \cite{Shukla1999}). Nonetheless, these methods do not account for how the market beta theoretically emanates from. Instead, these methods treat the market factor as an abstract object to be estimated and not truly real. Consequently, factor models are not explanatory, and can lead to (and have led to) serious misunderstanding of the market risk and the misuse of statistical tools to measure it. 

In this paper, we propose the \textit{Graphical Representation Model} (or \textit{GRM}) that aims at modeling the maximum endogenous variance of asset returns and that employs a sparse graphical model to characterize the association among assets. We make the following contributions:

\begin{enumerate}
\item {\bf The endogenous return perspective.}
We offer a new perspective on the problem of covariance 
estimation for asset returns. It suggests that much of the 
variance observed in equity markets may be captured
endogenously. This viewpoint differs from the one taken by  
the large and growing body of empirical and theoretical 
literature on factor models in finance. 
It posits that endogenous variables are sufficient in
explaining variation in the asset returns without the need for 
arbitrarily designed exogeneous variables; a theme that has 
dominated the financial literature. Figure \ref{fig:pve} shows that the proportion of variance explained by factors (via PCA 5-factor model) and by endogenous variables (via GRM, Graphical Representation Model, which will be introduced later in this paper) are highly similar across the past decades, for the S\&P 500 Index component stocks.
\begin{minipage}{\linewidth}
    \centering
    \includegraphics[scale = 0.63]{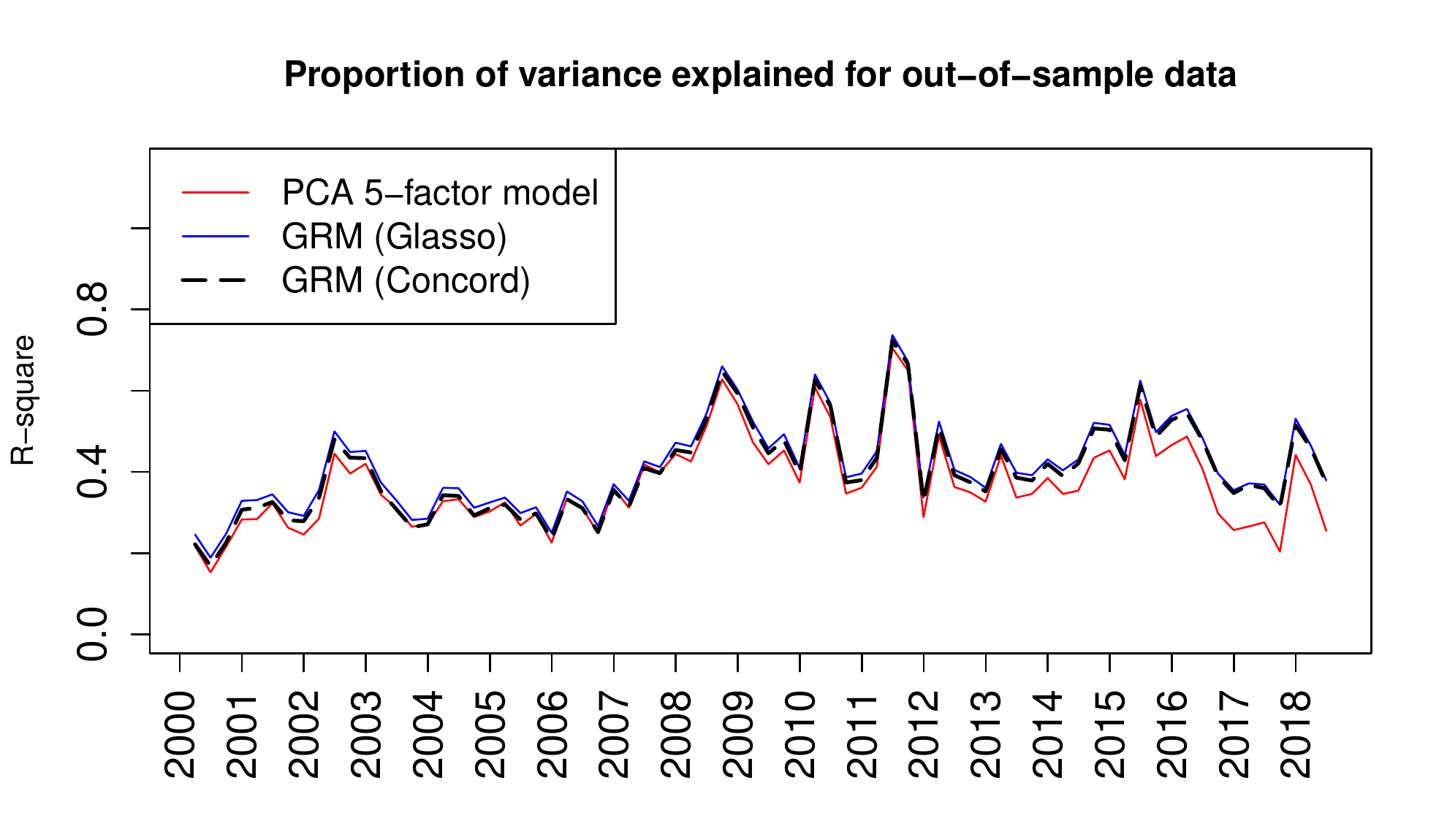}
    \captionof{figure}{A comparison of the PCA 5-factor model and GRM in terms of the proportion of variance explained (PVE), measured by $R^2$, for out-of-sample data using S\&P 500 component stocks. Please see the Appendix for more details.}
    \label{fig:pve}
\end{minipage}

To illustrate, we claim 
it suffices to estimate a set of regression coefficient $\{\rma_{ij}\}$
such that the return $Y_i$ to the $i$-th asset has
\begin{align} \label{model}
 Y_i = \sum_{j \neq i} \rma_{ij} Y_j + E_i
\end{align}
where the residual $E_i$ is uncorrelated with every return
$Y_j$ for which $i \neq j$. This perspective is closer in 
spirit to the CAPM which points to the use of an 
endogenous regressors. 
However, it does not
suffer from the bias these approaches generate due to their
requirement of a proxy portfolio (a linear combination of
returns as in $\req{model}$) as a regressor.
In particular, the CAMP
suggest estmating coefficients $\{ \beta_i \}$ such that 
the following residual is minimized
\begin{align}\label{capm}
    Y_i - \beta_i x{Y} 
  = Y_i - \sum_{j=1}^p \beta_i x_j Y_j
\end{align}
In $\req{capm}$ the $x$ is the market portfolio, which
is unknown and in practice $x{Y}$ is approximated
by the return to some market-like index, typically
the \rmS\&\rmP $500$. Notice that the residual $\req{capm}$
is necessarily correlated with the regressor $x{Y}$.
This leads to a bias in the regression procedure
alluded to above. In contrast, in equation 
$\req{model}$, no return $Y_i$ is explained 
in term of itself, and $\Cov (E_i, Y_j)$ for all
$j \neq i$.

A model of the return that has the representation 
$\req{model}$ has the 
advantages of being (A) potentially interpretable,
as the $\rma_{ij}$ can be parameterized by pairwise conditional
dependencies between assets, and 
(B) data driven, as the $\{\rma_{ij}\}$ may be estimated 
purely from observation, without the need for
the market index or additional factors. 
We demonstrate empirically on
several markets that endogenous approach explains more out of
sample variance than popular factor model approaches.

\item {\bf On graphical models and sparsity.}
The endogenous perspective described above is related
to graphical models which have popular applications
in numerous disciplines. However, graphical models have not
been a mainstream approach used in covariance estimation
for financial data sets. We suggest that the reason for 
this gap is the lack of two ingredients that would facilitate
their application: First, the formalism of graphical models makes them difficult 
to reconcile with the traditional ways
of thinking about returns to financial assets. These models
are also often coupled with Gaussian distributional 
assumptions which have been empirically invalidated. Second, graphical models are by convention associated with a
sparse precision matrix (i.e., the inverse of the covariance) and,
to our knowledge, there is no investigation in the 
literature of whether a sparsely estimated precision matrix
is consistent with well-understood properties 
of equity markets.

We address both of these shortcomings in our work. We explain
how the endogenous representation of returns in $\req{model}$
is related to the popularly adopted graphical models, and
related them to widely used factor models.
We also empirically test whether the assumption of sparsity
preserves the stylized properties of markets that are
commonly believed to hold. In particular, we focus on extracting
the market beta and the effects  of industry sectors (e.g. GICS, Global Industry Classification Standard)
from the coefficients $\{\rma_{ij}\}$.

\end{enumerate}

The rest of the paper is organized as follows. Section \ref{gcapm} introduces the endogenous representation of asset returns, the decomposition of asset variance, the graphical representation model, and the sparse model estimation. Section \ref{sparsity} provides comparisons with the CAPM, the factor modeling, and the mixed model, and discusses the implication of the model sparsity in the comparisons. Section \ref{empirical:grm} provides an empirical comparison of the out-of-sample performance of GRM, factor models, and mixed models, explores the properties of implied beta, and demonstrates the graph visualization and the interpretability of the detected stock communities. Section \ref{summary} summarizes the paper and discusses possible future work.

\section{The Model}
\label{gcapm}

\subsection{Endogenous representation of returns}

Suppose that there are $p$ risky assets in the financial market. Denote $Y = (Y_1, Y_2, \cdots, Y_p)^{\top}$, where $Y_i$ is the return of the $i$-th asset. Without loss of generality, we take every $Y_i$ to have zero mean. We posit that any security return $Y_i$ is well explained by the returns to the remaining securities. For non-random coefficients $\{\rma_{ij}\}$, we consider the endogenous representation of asset returns $Y$ as

\begin{equation}
\label{eq:lr}
    Y_i = \sum_{j\neq i}\rma_{ij}Y_j + E_i\quad\text{for}\quad i = 1, 2, \cdots, p
\end{equation}

\noindent where $E_i$ represents the residual return of the $i$-th asset. We further select the coefficients $\{\rma_{ij}\}$ such that $E_i$ is minimized for each $i = 1, 2, \cdots, p$. Such coefficients leads to the property that

\begin{equation}
    \label{eq:zerocov}
    \Cov(Y_j, E_i) = 0\quad\text{for every $i\neq j$}
\end{equation}
This property is appealing as it allows $Y_i$ in $\req{eq:lr}$ to be decomposed into an 
endogenous return $\sum_{j \neq i}\rma_{ij} Y_j$ and an uncorrelated (per $\req{eq:zerocov}$) residual. This parallels a common theme from linear models, in which a predictor and its residual are uncorrelated.\footnote{Note, however, that the residuals $\{E_i\}$ must be correlated amongst themselves, for otherwise, all the $\{Y_i\}$ are uncorrelated, yielding an unrealistic modeling assumption in our setting.} In the remaining of this paper, we make the following arguments: First, the coefficients $\{\rma_{ij}\}$ in $\req{eq:lr}$ that satisfy $\req{eq:zerocov}$ ensure that the total variance of the residual is minimized, or equivalently, the total variance in the endogenous component (in a sense to be made precise) is maximized; Second, the total variance of the return decomposes as the sum of the total variance of the endogenous return and the total variance of the residual; Third, the representation of $\req{eq:lr}$ has some similarities as well as distinctions to the CAPM one-factor model; Fourth, an assumption of an underlying graph (or network) structure that leads to sparse coefficients $\{\rma_{ij}\}$ remains consistent with a market model view of security returns\footnote{The market model is commonly attributed to the CAPM}, i.e., the model in which a systematic factor (i.e., market beta) drives the returns to all securities either up or down in unison.

The next formulation, also shown by \cite{Peng2009}, of a standard result sheds light on the coefficients $\{\rma_{ij}\}$ therein.

\begin{lemma} \label{lem:uncorrelated}
Suppose $Y\in\mathbb{R}^p$ has zero mean and positive definite covariance matrix $\bsig = \Var(Y)$ such that $\bomg = ((\omega_{ij}))= \bsig^{-1}$ exists. If $Y_i$ is expressed
\begin{equation}
\label{eq:endogenous}
    Y_i = \sum_{j\neq i}\rma_{ij}Y_j + E_i\quad\text{for}\quad i = 1, 2, \cdots, p
\end{equation}
where $a_{ii}=0$ and $\cov(Y_j, E_i) = 0$ for every $i\neq j$, then 
\begin{equation}
\label{eq:coef}
    \rma_{ij} = 
    \begin{cases}
    -\frac{\omega_{ij}}{\omega_{ii}}&,\quad\text{for $i\neq j$}\\
    0&,\quad\text{for $i=j$}
    \end{cases}
\end{equation}
\end{lemma}

Lemma \ref{lem:uncorrelated} provides a condition that the residual is uncorrelated with each of the predictors in the endogenous representation. The following Theorem \ref{lem:mini} offers a different condition and leads to the same solution for the coefficient matrix ${\bf A}$.

\begin{theorem} \label{lem:mini}
Suppose $Y\in\mathbb{R}^p$ has zero mean and positive definite covariance matrix $\bsig = \Var(Y)$ such that $\bomg = \bsig^{-1}$ exists. Then, there exists a unique minimizer ${\bf A}$:

\begin{equation}
\label{eq:resvar}
    {\bf A} = \argmin_{{\bf M}:{\bf M}_{ii}=0}\mathbb{E}(|Y - {\bf M} Y|^2)
\end{equation}
and the minimizer ${\bf A} = ((\rma_{ij}))= {\bf I} - {\bf D}\bomg$ where ${\bf D} = \diag(\bomg)^{-1}$. That is, if $\bomg = ((\omega_{ij}))$, then
\begin{equation}
\label{eq:results}
    \rma_{ij} = 
    \begin{cases}
    -\frac{\omega_{ij}}{\omega_{ii}}&,\text{ for $i\neq j$}\\
    0&,\text{ for $i=j$}
    \end{cases}
\end{equation}
\end{theorem}

\begin{remark} \label{rem:varmin}
It is easy to check that $\mathbb{E}(|Y - {\bf M} Y|^2)$ equals to $\text{tr}(\Var(Y - {\bf M} Y))$, which is further equal to the sum of the eigenvalues of the $p\times p$ covariance matrix $\Var(Y - {\bf M} Y)$. This sum represents the total variance over the principal components of this matrix.
\end{remark}

\begin{remark}
The restriction to zero of the diagonal entries of the elements in Equation (\ref{eq:resvar}) prevents the trivial minimizer ${\bf I}$ and is consistent with the endogenous representation where the response variable will not be used as one of the predictors.
\end{remark}


\subsection{The decomposition of variance}
\label{sec:variance}

Denoting by $\{\rma_{ij}\}$ the entries of this ${\bf A}$, Theorem $\ref{lem:mini}$ supplies a matrix ${\bf A}$ that minimizes (in the sense of Remark $\ref{rem:varmin}$) the residual $E = Y - {\bf A} Y$ over all matrices with zeros on the diagonal. Under the same definition of coefficients $\{\rma_{ij}\}$, Lemma \ref{lem:uncorrelated} states that the residual $E_i$ is uncorrelated with each of the predictor return $Y_j$. Thus the total variance of the return $Y_i$ can be decomposed into the total variance from the predictor returns, which we call the total variance of endogenous return, and the total variance of the residual. Equivalently, this states that ${\bf A} Y$ explains as much variance in the returns $Y$ as is possible endogenously. We write $\req{eq:lr}$ as
\begin{align} \label{eq:mendogi}
  Y = {\bf A} Y + E
\end{align}
for a vector of residuals $E = (E_1, \dots, E_p)^\top$ obeying $\req{eq:zerocov}$, so that $\Cov(E,Y)$ is the diagonal matrix ${\bf D}$ of Theorem $\ref{lem:mini}$, i.e., every security specific residual $E_i$ is correlated with that security return $Y_i$ only. We have the following proposition for the decomposition of variance.

\begin{proposition} \label{prop}
If the precision matrix $\bomg = ((\omega_{ij}))$ of Y is positive definite (i.e., $\omega_{ii} > 0$), and ${\bf A}$ is defined in (\ref{eq:results}), then the variance of $Y_i$ can be decomposed as
\begin{align} \label{decom}
   \sigma_{ii} = \Var({\bf A} Y)_{ii} + \frac{1}{\omega_{ii}}
   \end{align}
where $\Var({\bf A} Y)_{ii}$ is the endogenous variance of $Y_i$ and $1/\omega_{ii}$ is the variance of the residual. In addition, the covariance between $Y_i$ and $Y_j$ ($i\neq j$) can be decomposed as
\begin{align}
    \sigma_{ij} = \Var({\bf A} Y)_{ij} + \frac{-\omega_{ij}}{\omega_{ii}\omega_{jj}}
\end{align}
\end{proposition}
\begin{proof}
The result can be directly obtained by taking the $i$-th diagonal element and the $(i,j)$-th element of
\begin{equation}
    \Var({\bf A} Y) = {\bf A}\Var(Y){\bf A}^{\top} = ({\bf I} - {\bf D}\bomg)\bsig({\bf I} - {\bf D}\bomg)^{\top} = \bsig - 2{\bf D} + {\bf D}\bomg{\bf D}
\end{equation}
\end{proof}

We illustrate on a market with two assets the decomposition of the return $Y$ into its endogenous and residual parts.

\begin{example} \label{ex:two}
Let $Y = (Y_1, Y_2)^\top$ be the returns to $p = 2$ securities and take
\begin{align*}
  \Var(Y) = \bsig = \begin{pmatrix}
   \sigma_{11} & \sigma_{12} \\ \sigma_{12} & \sigma_{22}
   \end{pmatrix}
  = \bomg^{-1} = 
  \begin{pmatrix}
  \omega_{11} & \omega_{12}\\
  \omega_{12} & \omega_{22}
  \end{pmatrix}^{-1}
    = \frac{1}{|\bomg|}
    \begin{pmatrix}
    \omega_{22} & -\omega_{12}\\
    -\omega_{12} & \omega_{11}
   \end{pmatrix}.
\end{align*}
The least squared estimates for the coefficients $\rma_{12}, \rma_{21}$ in the system of linear equations
\begin{equation}
\begin{split}
    Y_1 = &\rma_{12}Y_2 + E_1\\
    Y_2 = &\rma_{21}Y_1 + E_2
    \end{split}
\end{equation}
are
\begin{equation}
    \rma_{12} = \frac{\Cov(Y_1, Y_2)}{\Var(Y_2)} = \frac{\sigma_{12}}{\sigma_{22}} = -\frac{\omega_{12}}{\omega_{11}},\quad
    \rma_{21} = \frac{\Cov(Y_2, Y_1)}{\Var(Y_1)} = \frac{\sigma_{12}}{\sigma_{11}} = -\frac{\omega_{12}}{\omega_{22}}.
\end{equation}
which yield
\begin{equation}
    \begin{pmatrix}
    Y_1 \\ Y_2
    \end{pmatrix}
    =
    \begin{pmatrix}
        0  & -\omega_{12}/\omega_{11}\\
        -\omega_{12}/\omega_{22} & 0
    \end{pmatrix}
    \begin{pmatrix}
    Y_1 \\ Y_2
    \end{pmatrix}
    +
    \begin{pmatrix}
    E_1 \\ E_2
    \end{pmatrix}.
\end{equation}
Thus $Y = {\bf A} Y + E$ with ${\bf A} = {\bf I} - {\bf D}\bomg$, ${\bf D} = \diag(\bomg)^{-1}$, and $E = (E_1, E_2)^{\top}$. In the endogenous component ${\bf A} Y$, the return $Y_1$ is explained in terms of the other return $Y_2$ and vice versa. Not surprisingly, in the case of $p = 2$, the correlation $\rho_{12} = \frac{\sigma_{12}}{\sqrt{\sigma_{11} \sigma_{22}}} = -\frac{\omega_{12}}{\sqrt{\omega_{11}\omega_{22}}}$ of the returns $Y_1$ and $Y_2$ is indicative of
how well they explain one another. For example, if the two security returns are uncorrelated (i.e., $\rho_{12} = 0$), neither one can account for the return of the other, and $E$ fully determines $Y$. We have,
\begin{align*}
  \Var(E) &= ({\bf I} - {\bf A}) \bsig ({\bf I} - {\bf A})^\top 
   = (1 - \rho_{12}^2)  
   \begin{pmatrix}
   \sigma_{11} & -\sigma_{12} 
\\ -\sigma_{12} & \sigma_{22}
   \end{pmatrix}
   = \begin{pmatrix}
   \frac{1}{\omega_{11}}  & \frac{\omega_{12}}{\omega_{11}\omega_{22}} \\
   \frac{\omega_{12}}{\omega_{11}\omega_{22}} & \frac{1}{\omega_{22}}
   \end{pmatrix}.
\end{align*}
For $Y_i$ ($i = 1, 2$), its total variance $\sigma_{ii}$ is decomposed into two parts: the variation explained by $Y_j (j\neq i)$ and the residual variation $1/\omega_{ii}$. When $|\rho_{12}|$ is close to one, the residual (total) variance is small and the explanation of the return $Y$ is almost entirely endogenous. The residuals $E_1$ and $E_2$ are oppositely correlated as $-\rho_{12}$ almost by definition of being not endogenous. The endogenous variance is the diagonal of $\Var({\bf A} Y)= \rho_{12}^2 \bsig$ mirroring the discussion of the residual.

In the simple setting of two securities, the variance of each return is a convex combination of the endogenous and residual variances. For instance, the variance $\sigma_{11}$ of security $1$ decomposes as 
\begin{equation}
    \sigma_{11} = \rho_{12}^2\sigma_{11} + \frac{1}{\omega_{11}}
\end{equation}
where $\rho_{12}^2\sigma_{11}$ is the variance explained by $Y_2$ and $1/\omega_{11}$ is the variance of the residual. The same holds for the variance $\sigma_{22}$ of security $2$. It turns out this decomposition may be generalized to an arbitrary number of assets through the concept of a partial covariance.

\end{example}

The endogenous representation of each asset return is expressed in terms of all the remaining asset returns. Such representation can be extended more generally where only parts of the asset returns are represented endogenously, conditional on all the remaining asset returns. We have the following result.

\begin{theorem} \label{thm:partition}
Let $\mathbb{I}, \mathbb{J}\subseteq \{1,2,\cdots, p\}$ such that $\mathbb{I}\cap\mathbb{J} = \emptyset$ and $\mathbb{I}\cup\mathbb{J} = \{1,2,\cdots,p\}$. Suppose $Y\in\mathbb{R}^p$ has zero mean and positive definite covariance matrix $\bsig = \var(Y)$ such that $\bomg = \bsig^{-1}$ exists. Let $Y = (Y_{\mathbb{I}}, Y_{\mathbb{J}})^{\top}$ with mean and covariance matrix (without loss of generality, we use the pair $(1,2)$ to ease the notation.)
\begin{equation} \label{eq:sigma-partition}
    \mu = \begin{pmatrix}
    \mu_{\mathbb{I}}\\
    \mu_{\mathbb{J}}
    \end{pmatrix}\quad\text{and}\quad
    \bsig = \begin{pmatrix}
    \bsig_{\mathbb{I}\mathbb{I}} & \bsig_{\mathbb{I}\mathbb{J}}\\
    \bsig_{\mathbb{J}\mathbb{I}} & \bsig_{\mathbb{J}\mathbb{J}}
    \end{pmatrix}
\end{equation}
Define $Y_{\mathbb{I}|\mathbb{J}} = Y_{\mathbb{I}} - \mathcal{L} (Y_{\mathbb{J}})$, where $\mathcal{L} (Y_{\mathbb{J}}) = \bsig_{\mathbb{IJ}} \bsig_{\mathbb{JJ}}^{-1} Y_{\mathbb{J}}$, then in the endogenous representation 
\begin{equation}
\label{eq:38}
    Y_{\mathbb{I}|\mathbb{J}} = {\bf A}_{\mathbb{I}|\mathbb{J}}Y_{\mathbb{I}|\mathbb{J}} + E_{\mathbb{I}|\mathbb{J}}
\end{equation}
where ${\bf A}_{\mathbb{I}|\mathbb{J}}$ is a $|\mathbb{I}|\times|\mathbb{I}|$ matrix with $({\bf A}_{\mathbb{I}|\mathbb{J}})_{ii} = 0$ and $E_{\mathbb{I}|\mathbb{J}}$ is a $|\mathbb{I}|$-dimensional vector, such that
\begin{equation}
\label{eq:39}
    \Cov((Y_{\mathbb{I}|\mathbb{J}})_i, (E_{\mathbb{I}|\mathbb{J}})_j) = 0
    \quad\text{for all $i\neq j$}
\end{equation}
then
\begin{equation}
    {\bf A}_{\mathbb{I}|\mathbb{J}} = {\bf A}_{\mathbb{I}\mathbb{I}}
\end{equation}
where ${\bf A} = {\bf I} - {\bf D}\bomg, {\bf D}=\diag(\bomg)^{-1}$.
\end{theorem}

Following the Example \ref{ex:two}, we partition the returns as $Y = (Y_{\mathbb{I}}, Y_{\mathbb{J}})^{\top}$ where 
$\mathbb{I} = \{1, 2\}$ and $\mathbb{J} = \{3, \cdots, p\}$. This leads to a partitioning of $\bsig$ as in the Equation (\ref{eq:sigma-partition}). The partial correlation $\varrho$ and partial covariance $\pi$ of $(Y_1, Y_2)$ are defined as
\begin{align} \label{eq:pcov}
    \varrho_{ij} = \frac{\pi_{ij}}{\sqrt{\pi_{ii}\pi_{jj}}}
      \hspace{0.16in} \text{ and} \hspace{0.16in}
    \pi_{ij} = \Cov (Y_i - \mathcal{L}_i(Y_{\mathbb{J}}), Y_j - \mathcal{L}_j(Y_{\mathbb{J}})) 
\end{align}
where $\mathcal{L} (Y_{\mathbb{J}}) = \bsig_{\mathbb{IJ}} \bsig_{\mathbb{JJ}}^{-1} Y_{\mathbb{J}}$ is the best affine (in the least-squares sense) estimator of $Y_{\mathbb{I}}$ in terms of $Y_{\mathbb{J}}$ and the $i,j \in \{1,2\}$.\footnote{This estimator is distinct from the conditional expectation $\mathbb{E}(Y_{\mathbb{I}} \bst Y_{\mathbb{J}})$ which is the best least-squares estimator of $Y_{\mathbb{I}}$ that is $Y_{\mathbb{J}}$-measurable. The definition applies to any pair $(i,j)$ provided
we set $Y_{\mathbb{I}} = (Y_i, Y_j)$ and let $Y_{\mathbb{J}}$ be the vector of $p-2$ elements $Y_k$ for $k \notin \{i,j\}$.} 
The partial covariance $\pi_{12}$ is the covariance between $Y_1$ and $Y_2$ after removing all first order (linear) dependence on the remaining variables in $Y_{\mathbb{J}}$. In this sense, $\varrho_{12}$ is perhaps a better measure of pairwise correlation.

Our next result provides the decomposition of the variance of a security return into its endogenous and residual parts. It generalizes Example $\ref{ex:two}$.

\begin{proposition}
Suppose $Y\in\mathbb{R}^p$ has zero mean and positive definite covariance matrix $\bsig=((\sigma_{ij}))$. Let $\bomg=((\omega_{ij})) = \bsig^{-1}$, and $Y = (Y_{\mathbb{I}}, Y_{\mathbb{J}})^{\top}$ where 
$\mathbb{I} = \{1, 2\}$ and $\mathbb{J} = \{3, \cdots, p\}$. Denote $\nu_{11} = \Var({\bf A} Y)_{11}$, the endogenous variance of $Y_1$. Then,
\begin{align} \label{decom2}
   \nu_{11} =  \sigma_{11} - \pi_{11} (1 - \varrho_{12}^2)
\end{align}
where $\pi_{11}$ and $\varrho_{12}$ are defined in (\ref{eq:pcov}).
\end{proposition}
The result of $\req{decom2}$ states that the total variance of $Y_1$ is decomposed into two parts: $\nu_{11}$ and $\pi_{11} (1 - \varrho_{12}^2)$. $\nu_{11}$ is the endogenous variance of $Y_1$, that is the variance of $Y_1$ explained by all other assets $Y_2, Y_3, \cdots, Y_p$ in the system; $\pi_{11} (1 - \varrho_{12}^2)$ is the residual variance that is specific to $Y_1$ and cannot be explained by the system.
\begin{remark}
 It may appear that $\nu_{11}$ depends on the choice of $Y_2$ unless $\varrho_{1j}$ is the same for all $j$. However, this is not so because $\pi_{11}$ depends on the choice of $Y_2$.` Nevertheless, the choice of $Y_2$ is arbitrary in the computation of $\nu_{11}$.
\end{remark}

We proceed to explain analogy to Example $\ref{ex:two}$. Denote by $\bpi$ the partial 
covariance matrix of the pair $(Y_1,Y_2)$. By $\req{eq:pcov}$,
\begin{align}
   \bpi = \begin{pmatrix}
   \pi_{11} & \pi_{12}  \\ 
   \pi_{12} & \pi_{22} 
   \end{pmatrix}
   = \bsig_{\mathbb{I}\mathbb{I}} - \bsig_{\mathbb{I}\mathbb{J}}\bsig_{\mathbb{J}\mathbb{J}}^{-1}\bsig_{\mathbb{J}\mathbb{I}}
\end{align}
In analogy with Example $\ref{ex:two}$, the concentration matrix $\bomg = \bsig^{-1}$, which admits a partition akin to $\req{eq:sigma-partition}$, has the following entries corresponding to the variables in $Y_{\mathbb{I}}$
\begin{align}
  \bomg_{VV} = \frac{1}{|\bpi|}  
  \begin{pmatrix}
   \pi_{22} & -\pi_{12}  \\ 
   -\pi_{12} & \pi_{11} 
   \end{pmatrix}.
\end{align}
Using the inverse of the diagonal entries $|\bpi| / \pi_{ii} = (1 - \varrho_{ij}^2) \pi_{jj}$, we obtain the corresponding entries of the matrix ${\bf A}$ of Lemma $\ref{lem:mini}$. For the variables $V$,
\begin{align}
  {\bf A}_{VV} =
  \begin{pmatrix}
    0 & \pi_{12} / \pi_{22}  \\ 
   \pi_{12}/\pi_{11} & 0
   \end{pmatrix}
\end{align}
The residual matrix $\bdel = \Var(E)$ for $V$ may then be computed as
\begin{align}
  \bdel_{VV} =  (1 - \varrho_{12}^2) 
   \begin{pmatrix}
   \pi_{11} & -\pi_{12} 
\\ -\pi_{12} &  \pi_{22}
   \end{pmatrix}.
\end{align}
The endogenous variance corresponding to the variables $V$ only takes the form
$\Var({\bf A}_{VV} V) = \varrho_{12}^2 \bpi$ that is analogous to Example $\ref{ex:two}$ with $\bpi$
replacing $\bsig$. The full endogenous variance matrix $\Var( {\bf A} Y)$ is computed as follows
\begin{align}
  \Var( {\bf A} Y) = {\bf A} \bsig {\bf A}^\top = \bsig - 2 {\bf D} + \bdel
\end{align}

\subsection{The graphical representation model}

Based on the endogenous representation of asset returns in $\req{eq:mendogi}$, we introduce the \textit{Graphical Representation Model} (or \textit{GRM}) for the asset returns $Y$ as
\begin{equation}
    Y = {\bf A} Y + E
\end{equation}
where ${\bf A} = {\bf I} - {\bf D}\bomg, {\bf D} = \diag(\bomg)^{-1}$, and ${\bf A}$ is a sparse matrix. In the GRM, we claim ${\bf A}$ to be sparse for two reasons. First, from a modeling perspective, we raise the possibility that the true precision matrix $\bomg$ is sparse. That is, the true partial correlations (therefore the direct interactions between firms) are sparse. Such a claim is based on the empirical observation that the correlations between stocks are largely due to the strong correlation between each stock and the market (e.g. equity indices), rather than inter-stock dependencies \cite{Shapira2009}. Thus after removing the mediating effect of other stocks, there are only a few significant stock correlations while others are negligible in comparison. Second, a sparse estimation of $\bomg$ (equivalently ${\bf A}$) stabilizes the estimated parameters when the sample size is small while also performing data-driven model selection to prevent over-fitting.

\subsection{Sparse model estimation}\label{estimation}

This sub-section discusses how to estimate ${\bf A}$ in the graphical representation model. To estimate ${\bf A}$, one needs to estimate the precision matrix (also called the concentration matrix) $\bomg$. A sparse structure of $\bomg$ is of great interest in many applications, where the sparsity pattern of the precision matrix directly corresponds to the graphical model structure. We assume that $\bomg$ is sparse for several reasons. Assuming sparsity in the estimation of $\bomg$ provides better interpretability of GRM and computational storage efficiency. Besides, from an estimation perspective, assuming sparsity allows us to use sparsity-included regularization that stabilizes estimation. Several graphical models have been proposed to estimate a sparse $\bomg$, include Graphical Lasso (Glasso)\cite{Friedman2008} and Concord\cite{Khare2014}:

\begin{equation}
    \begin{split}
        \text{Glasso:}&\quad\hat\bomg = \argmin_{\bthe\succeq0}\{-\log\det(\bthe) + \tr({\bf S}\bthe) + \lambda\|\bthe\|_1\}\\
        \text{Concord:}&\quad\hat\bomg = \argmin_{\bthe\in\mathbb{R}^{p\times p}}\{-\log\det(\bthe_D^2) + \tr({\bf S}\bthe^2) + \lambda\|\bthe_F\|_1\}
    \end{split}
\end{equation}

\noindent where $\bthe_D$ is the diagonal matrix of $\bthe$, $\bthe_F = \bthe - \bthe_D$, ${\bf S}$ is the sample covariance matrix defined as:

\begin{equation}
    {\bf S} = \frac{1}{n}\sum_{i=1}^n(y_i - \bar y)(y_i - \bar y)^{\top}\quad\text{where}\quad
    \bar y = \frac{1}{n}\sum_{i = 1}^ny_i
\end{equation}

\noindent for the given $n$ independent samples $y_1, y_2, \cdots, y_n\in\mathbb{R}^p$, and $\lambda$ is the tuning parameter that controls the sparsity level of $\hat\bomg$. Generally, the larger the value of $\lambda$, the more zero elements $\hat\bomg$ contains. In practice, $\lambda$ is determined by the cross-validation procedure\footnote{Cross-validation is a widely-used method for estimating the model's parameter. In $f$-fold cross-validation, observed data is partitioned into $f$ roughly equal-size parts. Given a certain value of $\lambda$, for the $i$-th part of the data, we use all remaining parts to fit the model and calculate the prediction error for the $i$-th part of data. We do this for $i = 1, 2, \cdots, f$ and combine the $f$ estimates of prediction error. The optimal $\lambda$ is the one that has the least prediction error. See more details in Section 7.10 of the book \textit{The Elements of Statistical Learning} \cite{Hastie2009}.}. Glasso is solved by the block coordinate descent algorithm \cite{Friedman2008} and can be implemented using the \textbf{R} package \texttt{glasso} \footnote{Glasso package can be downloaded from \url{https://cran.r-project.org/web/packages/glasso/index.html}.}. Concord can be solved either by the coordinate-wise descent algorithm \cite{Khare2014} or by the proximal gradient method \cite{Oh2014}. Both methods can be implemented in the \textbf{R} package \texttt{gconcord}. The estimated ${\bf A}$ is:

\begin{equation}
    \hat{{\bf A}} = {\bf I} - \hat{\bf D}\hat\bomg
\end{equation}
where $\hat{\bf D} = \diag(\hat\bomg)^{-1}$. A sparse $\hat\bomg$ returns a sparse estimation for ${\bf A}$, and both $\hat\bomg$ and $\hat{{\bf A}}$ have the same sparsity pattern, i.e. both matrices have the same locations of zero elements. The sparse estimation of ${\bf A}$ has two desirable properties: First, it implies a parsimonious model and allows one to visualize a sparse graph. Such a sparse graph provides a better interpretation of asset returns. At the same time, the estimated implied beta is still dense, thus the sparsity does not jeopardize the beta structure. Second, a sparse $\hat{{\bf A}}$ implies that in the linear regression $\req{eq:lr}$, many estimated regression coefficients are zeros. From this aspect, the estimation of a sparse ${\bf A}$ is similar to a Lasso regression for each asset return, and the sparsity estimation for ${\bf A}$ is equivalent to the variable selection in $\req{eq:lr}$, and enhances the prediction accuracy of the graphical representation model.

\section{A comparison with existing models}
\label{sparsity}

In this section, we compare the GRM with some major models that have dominated in the finance literature. We demonstrate that (A) the GRM shares some common properties with the one-factor model implied by the CAPM equation while distinction also exists, (B) despite that no factors are contained in the GRM, it is possible to obtain implied factors from it, and (C) adding additional factors in GRM conflicts the model sparsity.

\subsection{A comparison with CAPM one-factor model}
In this part, we demonstrate that both similarities and distinction exist between the GRM and the CAPM. Suppose that the risk-free interest rate is $R_f$ and the return of the market portfolio is $R_m$. Let $e = (1, 1, \cdots, 1)^{\top}\in\mathbb{R}^p$ and $w_m\in\mathbb{R}^p$ be the allocation of the market portfolio among the $p$ numbers of assets. The CAPM equation is:

\begin{equation}
    \label{eq:capm}
    \mathbb{E}(R) - R_fe = \beta(\mathbb{E}(R_m) - R_f)
\end{equation}

\noindent where $\beta = (\beta_1, \beta_2, \cdots, \beta_p)^{\top}$ and $\beta_i = \frac{\Cov(R_i, R_m)}{\Var(R_m)}$ is the beta of the $i$-th asset. \cite{Blume1971} empirically shows that all betas have the tendency towards one, and that almost all betas are positive, i.e., when the market goes up, all the returns go up, and when markets go down, the opposite is true. Define $Z_{\beta} = (R-R_fe) - \beta(R_m - R_f)$. Then $Z_{\beta}$ is interpreted as the vector of diversifiable specific returns of asset, and:

\begin{equation}
\begin{split}
    Z_{\beta} &= R - \mathbb{E}(R) - \beta(R_m - \mathbb{E}(R_m)) + (\mathbb{E}(R) - R_fe) - \beta(\mathbb{E}(R_m) - R_f)\\
    &= R - \mathbb{E}(R) - \beta(R_m - \mathbb{E}(R_m))
    \end{split}
\end{equation}
Denote the excess return of the market portfolio as $X = R_m - \mathbb{E}(R_m)$. The CAPM equation implies a one-factor model:

\begin{equation}
\label{eq:1-factor}
    Y = \beta X + Z_{\beta}
\end{equation}

\noindent In the one-factor model, the market beta $\beta$ is the regression coefficient for $Y$ against $X$. This is different from the GRM where $\{\rma_{ij}\}$ are the regression coefficients for $Y_i$ against $Y_j$ with all $j\neq i$. Since in CAPM the market portfolio $w_m\in\mathbb{R}^p$ is a value-weighted portfolio of all the securities in the market, we have

\begin{equation}
\label{market-factor}
    X = R_m - \mathbb{E}(R_m) = w_m^{\top}R - \mathbb{E}(w_m^{\top}R) = w_m^{\top}Y
\end{equation}

\noindent This implies that the market factor in CAPM is latent and is endogenously given, and that the one-factor model in $\req{eq:1-factor}$ can be written as

\begin{equation}
\label{eq:1.1-factor}
    Y = {\bf A}_{\beta}Y + Z_{\beta}\quad\text{where}\quad {\bf A}_{\beta} = \beta w_m^{\top}
\end{equation}
The equation $\req{eq:1.1-factor}$ has a similar representation with $\req{eq:mendogi}$ and it implies that the market beta is endogenously determined by the returns of all the assets and is the invisible hand that is the main driver behind the equity markets. In the one-factor model $\req{eq:1-factor}$, the return of each asset is described by a linear combination of all asset returns plus the specific returns. Therefore, CAPM treats $\beta$ as an abstract object that is not truly real, and the meaning of $\beta$ becomes elusive and incomprehensible. As a result, the factor model in $\req{eq:1-factor}$ is not explanatory and can lead to a serious misunderstanding of market risk and the misuse of statistical tools to measure it. 

Furthermore, the equation $\req{eq:1.1-factor}$ also unveils that the covariance matrix of $Z_{\beta}$ cannot be diagonal, as demonstrated in Lemma \ref{lemma1}.

\begin{lemma}
\label{lemma1}
In the one-factor model $Y = \beta X + Z_{\beta}$ where $X = w_m^{\top}Y$ is the market factor and $w_m$ is the market portfolio, the covariance matrix $\bdel_{\beta} = \Var(Z_{\beta})$ is not diagonal.
\end{lemma}

\begin{proof}
Let $(w_m)_i$ be the $i$-th element of $w_m$. Then:

\begin{equation*}
\begin{split}
    w_m^{\top}\beta &= \sum_{i = 1}^p(w_m)_i\beta_i = \sum_{i=1}^p(w_m)_i\cdot\frac{\Cov(Y_i, X)}{\Var(X)}\\
    &= \frac{1}{\Var(X)}\Cov\Big(\sum_{i=1}^p(w_m)_iR_i, X\Big) = \frac{\Cov(X, X)}{\Var(X)} = 1
    \end{split}
\end{equation*}
That is, the market portfolio has a beta of 1. Left multiply $w_m^{\top}$ on both sides of the one-factor model, we have:

\begin{equation*}
w_m^{\top}Z_{\beta} = w_m^{\top}Y - w_m^{\top}\beta X = w_m^{\top}Y - X = 0
\end{equation*}

\noindent Assume that $\bdel_{\beta}$ is diagonal with non-zero diagonal elements, and $Z_i$ be the $i$-th element of $Z_{\beta}$, then $\mathbb{E}(Z_iZ_j) = 0$ for all $i\neq j$. Thus:

\begin{equation*}
 w_m^{\top}\bdel_{\beta} = w_m^{\top}\mathbb{E}(Z_{\beta}Z_{\beta}^{\top}) = \mathbb{E}(w_m^{\top}Z_{\beta}Z_{\beta}^{\top}) = \mathbf{0}  
\end{equation*}

\noindent Because $w_m\neq \mathbf{0}$, we have $\bdel_{\beta} = \mathbf{0}$, this is a contradiction. Thus $\bdel_{\beta}$ is not diagonal.
\end{proof}

\noindent Consequently, what CAPM has taken literally means is that the specific returns of assets are correlated. Such a fact is consistent with the result that in the GRM the covariance matrix of $E$ is non-diagonal. That is

\begin{equation}
    \label{eq:cov}
    \Var(E) = {\bf D}\bomg{\bf D}
\end{equation}

\noindent This is different from a factor model whose residual is generally assumed to have a diagonal covariance matrix. Unfortunately, the non-diagonal fact of the covariance matrix of $Z_{\beta}$ has been ignored by many factor models, especially the strict factor modeling whose residual term $Z_{\beta}$ is generally assumed to have a diagonal covariance matrix.

\subsection{A comparison with multi-factor modeling}

In this sub-section, we try to understand the connection between the GRM and a factor model. There has been a whole different variety of factor models in the literature\cite{Litterman2003}. CAPM is the origin of the one-factor model. As an alternative, Arbitrage Pricing Theory (APT) \cite{Ross1976} introduces a multi-factor model. Since then, the literature has gone in different ways, including fundamental factor models (such as the Fama-French factor model that identifies specific factors in an equity market) and statistical factor models (such as PCA (Principal Component Analysis) factor models that take principal components as factors). 

Consider a $k$-factor model:

\begin{equation}
\label{eq:apt}
    Y = {\bf B} X + Z
\end{equation}

\noindent where ${\bf B}$ is a $p\times k$ matrix of factor exposures and $X$ is a $k$-dimensional vector of zero-mean common factor returns. Assume that $\mathbb{E}(XZ^{\top}) = \mathbf{0}$, then the covariance matrix of $Y$ given by the factor model is:

\begin{equation}
    \bsig_{FM} = {\bf B}{\bf V}{\bf B}^{\top} + \bdel_{FM}
\end{equation}

\noindent with ${\bf V} = \Var(X)$ and $\bdel_{FM} = \Var(Z)$. By Woodbury identity \cite{Woodbury1950}, it can be shown that the precision matrix of the factor model can be decomposed as:

\begin{equation}
    \bomg_{FM} = \bsig_{FM}^{-1} = \bdel_{FM}^{-1} + {\bf L}
\end{equation}

\noindent where ${\bf L}$ is a low-rank and dense matrix\footnote{By Woodbury identity ${\bf L} = -\bdel_{FM}^{-1}{\bf B}({\bf V}^{-1} + {\bf B}^{\top}\bdel_{FM}^{-1}{\bf B})^{-1}{\bf B}^{\top}\bdel_{FM}^{-1}$ and $\text{rank}({\bf L}) = k$.}. Thus, in a factor model, $\bomg_{FM}$ must be a dense matrix. A detailed discussion can be found in \cite{Shkolnik2016}.

In contrast to a factor model, the GRM allows the possibility that $\bomg$ to be a sparse matrix. Despite that the GRM does not contain factors, it is possible to obtain implied factors. Suppose that in the factor model of $\req{eq:apt}$, ${\bf B} = (\beta_1, \beta_2, \cdots, \beta_k)$ where each $\beta_j$ is a $p$-dimensional column vector of $\bf B$, and the $k$ leading eigenvectors of $\bsig_{FM}$ are $\xi_1, \xi_2, \cdots, \xi_k$, then by \cite{Chamberlain1983}, as $p\rightarrow\infty$, we have $\beta_i\rightarrow\xi_i$ for $i = 1, 2, \cdots, k$. Motivated by this result, we define the implied factor matrix of the GRM as follows. Suppose the $k$ leading eigenvectors of $\bomg^{-1}$ are $\beta^{imp}_1, \beta^{imp}_2, \cdots, \beta^{imp}_k$, then the $p\times k$ implied factor matrix is defined as:

\begin{equation}
    {\bf B}^{imp} \triangleq (\beta^{imp}_1, \beta^{imp}_2, \cdots, \beta^{imp}_k)
\end{equation}

\noindent In particular, in the one-factor model, we take $\beta_1$ as the market beta. In the GRM, we take $\beta^{imp}_1$ as the implied beta. One interesting question is: if $\bomg$ is actually sparse, how is ${\bf B}^{imp}$ compared with ${\bf B}$, or in particular, how is $\beta^{imp}_1$ compared with $\beta_1$? The empirical answer will be unveiled in Section \ref{implied-beta}.

\subsection{A consideration for the mixed model}
\label{section5.2}

In spatial econometrics, the spatial interaction model emphasizes the statistical modeling of spatial interaction. For example, the spatial interaction model proposed by \cite{Billio2015} and \cite{Kou2018} is
\begin{equation}
    Y = \rho{\bf W} Y + {\bf B} X + G
\end{equation}
where $\rho$ is a scalar parameter to be estimated, ${\bf W}$ is the spatial interaction term measured by the geographical distances among assets, $X$ is the vector of factor exposures with the factor loading matrix ${\bf B}$, and $G$ is the residual. Motivated by the spatial interaction model, we are curious about the performance of GRM where factors are added into the model. That is, we consider the mixed model:
\begin{equation}
\label{eq:spatial}
Y = \rho{\bf A} Y + {\bf B} X + G
\end{equation}
where $\rho$ is a scalar parameter to be estimated, $G$ is the residual and ${\bf B}$ is the coefficient matrix (different from the ${\bf B}$ in a factor model of $\req{eq:apt}$), ${\bf A}=((a_{ij}))$ is defined in Equation (\ref{eq:results}). The mixed model contains the graphical association confounded with factors. In the mixed model, we are interested in whether adding the additional factors $X$ would further improve the out-of-sample performance of GRM.

Our empirical study in Section \ref{section:5.1} shows that the graphical representation model is useful to characterize asset returns, and when ${\bf A} Y$ term is present, factors $X$ become redundant. In addition, from a modeling perspective, such a result has an interesting implication. \cite{Shkolnik2016} shows that in a factor model, the precision matrix $\bomg$ can be decomposed into a sparse matrix and a dense low-rank matrix by the Sparse Low-rank Decomposition (SLD) proposed by \cite{Chandrasekaran2012}, and in the setup of a factor model $\bomg$ is not sparse. That is, the mixed model $\req{eq:spatial}$ implies that the precision matrix of $Y$ can be decomposed as
\begin{equation}
    \bomg = ({\bf I} - \rho{\bf A})^{\top}({\bf B}{\bf V}{\bf B}^{\top} + \bdel_{G})^{-1}({\bf I} - \rho{\bf A})
\end{equation}
where ${\bf V} = \Var(X)$ and $\bdel_G = \Var(G)$. Although $\bdel_G$ is a diagonal matrix, ${\bf B}{\bf V}{\bf B}^{\top}$ is dense and has a low rank (with the rank equal to the total number of orthogonal factors) thus $\bomg$ is a dense matrix. However, when factor terms are not present (i.e., the term ${\bf B}{\bf V}{\bf B}^{\top}$ is gone), the low-rank part goes away and $\bomg$ becomes sparse. From this point of view, a sparse precision matrix $\bomg$ is consistent with the GRM, not with the mixed model.

\section{Results on empirical data}
\label{empirical:grm}

Our empirical study shows that (A) compared with factor models, GRM is useful to characterize the asset returns, (B) the implied beta from GRM has all the properties of betas that come out from the factor models, and (C) GRM allows graph visualization among assets thus provide straightforward interpretation for assets inter-dependencies. While factors are invisible to financial market practitioners, the association among assets can be fully visualized.

In this section, we use the following notations: $n_I$ and $n_O$ are the sample sizes for in-sample and out-of-sample data. Let ${\bf Y}_I$ be the $p\times n_I$ data matrix of in-sample asset returns and ${\bf Y}_O$ be the $p\times n_O$ data matrix of out-of-sample asset returns, and $\hat{{\bf Y}}_O$ be the predicted values for ${\bf Y}_O$. Let ${\bf X}_I$ and ${\bf X}_O$ be the $k\times n_I$ and $k\times n_O$ data matrices of in-sample and out-of-sample observations for $k$ factors, respectively. For an arbitrary matrix ${\bf M}$, the $(i,j)$-th element is denoted as ${\bf M}_{i,j}$.

\subsection{Out-of-sample performance} \label{section:5.1}

To show that a sparse matrix ${\bf A}$ is sufficient to model asset returns, we compare the out-of-sample performance of the following models: the Fama-French factor model, the PCA factor model, the spatial interaction model proposed by \cite{Kou2018}, the graphical representation model, and the mixed model. We download daily close prices for all component stocks of S\&P 500 Index and Dow Jones Industrial Average Index (DJIA). We compute stock returns and split the dataset into in-sample and out-of-sample data whose time periods are presented in Table \ref{tab:time}\footnote{January 2008 was regarded as the beginning of the crisis, when a major mortgage lender defaulted and negative news started swirling. In March 2008, Bear Stearns was sold to JP Morgan for almost nothing with major Federal intervention.}. For each set of in-sample and out-of-sample data, the data set of asset returns is centered to have zero mean and is denoted as ${\bf Y}_I$ and ${\bf Y}_O$, respectively. For each model, we follow the below steps to compute the out-of-sample prediction $\hat{{\bf Y}}_O$:

\begin{itemize}
    \item Fama-French $k$-factor\footnote{We consider Fama-French 3-factor and 5-factor models. 3 factors are market, SMB  (Small Minus Big), HML (High Minus Low), 5 factors are market, SMB, HML, RMW (Robust Minus Weak), CMA (Conservative Minus Aggressive). See \cite{Fama1993} and \cite{Fama2015}.} model $Y = {\bf B} X + \zf$:
    \begin{enumerate}
        \item Download Fama-French factor returns from Kenneth R. French Data Library\footnote{See \url{https://mba.tuck.dartmouth.edu/pages/faculty/ken.french/Data_Library/f-f_5_factors_2x3.html}} and set as ${\bf X}_I$ and ${\bf X}_O$ by Table \ref{tab:time}. Each column of ${\bf X}_I$ and ${\bf X}_O$ is centered to have zero mean;
        \item The least square estimate for ${\bf B}$ is $\hat{{\bf B}} = {\bf Y}_I{\bf X}_I^{\top}({\bf X}_I {\bf X}_I^{\top})^{-1}$;
        \item The predicted value for ${\bf Y}_O$ is $\hat{{\bf Y}}_O = \hat{{\bf B}}{\bf X}_O$.
    \end{enumerate}
    
    \item PCA $k$-factor model $Y = {\bf B} X + \zf$:
    \begin{enumerate}
        \item Compute the sample covariance matrix ${\bf S} = \frac{1}{n_I-1}{\bf Y}_I{\bf Y}_I^{\top}$ and its $k$ largest eigenvectors $\hat\beta_1, \hat\beta_2, \cdots, \hat\beta_K$. Set $\hat{{\bf B}} = (\hat\beta_1, \hat\beta_2, \cdots, \hat\beta_K)$;
        \item The recovered out-of-sample latent factor returns $\hat{{\bf X}}_O = (\hat{{\bf B}}^{\top}\hat{{\bf B}})^{-1}\hat{{\bf B}}^{\top}{\bf Y}_O$;
        \item The predicted value is $\hat{{\bf Y}}_O = \hat{{\bf B}}\hat{{\bf X}}_O$.
    \end{enumerate}
    
    \item Spatial interaction model $Y = \rho {\bf W} Y + {\bf B} X + G$:
    \begin{enumerate}
        \item ${\bf W}$ represents the spatial distance of assets, ${\bf W}_{i,i} = 0$ and ${\bf W}_{i,j} = (s_id_{ij})^{-1}$ for $i\neq j$, where $d_{ij}$ is the (geographic or driving) distance\footnote{In our empirical study, the geographical distance between two assets is measured as the spherical distance (in miles) of the earth between the headquarters of the two companies.} between asset $i$ and asset $j$, and $s_i\triangleq\sum_{j}d_{ij}^{-1}$.
        
        \item Obtain the least square estimate for $\rho$ by solving\footnote{
    Given a certain $\rho$ such that $I - \rho{\bf W}$ is invertible, we have $({\bf I}-\rho{\bf W}){\bf Y}_I = {\bf B}{\bf X}_I + \zo$, the least square estimate for ${\bf B}$ is $\hat{{\bf B}} = ({\bf I} - \rho{\bf W}){\bf Y}_I{\bf X}_I^{\top}({\bf X}_I{\bf X}_I^{\top})^{-1}$, and the residual is
    \[
    ({\bf I} - \rho{\bf W}){\bf Y}_I - ({\bf I} - \rho{\bf W}){\bf Y}_I{\bf X}_I^{\top}({\bf X}_I{\bf X}_I^{\top})^{-1}{\bf X}_I
    = ({\bf I} - \rho{\bf W}){\bf Y}_I({\bf I} - {\bf X}_I^{\top}({\bf X}_I{\bf X}_I^{\top})^{-1}{\bf X}_I)
    \]
    The least square estimate for $\rho$ is the estimate that minimize the residual. The set of candidate values is $\mathcal{S} = \{\zeta\leq\kappa\leq\eta: {\bf I} - \kappa{\bf W}\text{ is invertible.}\}$.}:
    \begin{equation*}
        \hat\rho = \argmin_{\kappa\in\mathcal{S}}\quad\big\|({\bf I} - \kappa{\bf W}){\bf Y}_I({\bf I} - {\bf X}_I^{\top}({\bf X}_I{\bf X}_I^{\top})^{-1}{\bf X}_I)\big\|_F^2
    \end{equation*}
    
    \item The predicted value is $\hat{{\bf Y}}_O$ whose $(i,j)$-th element is:
    \begin{equation*}
        (\hat{{\bf Y}}_O)_{i,j} = \hat\rho\sum_{k\neq i}{\bf W}_{i,k}({\bf Y}_O)_{k,j} + \sum_{h = 1}^K\hat{{\bf B}}_{i,h}({\bf X}_O)_{h,j}
    \end{equation*}
    \end{enumerate}

    \item Graphical representation model $Y = {\bf A} Y + \zg$:
    \begin{enumerate}
        \item Given ${\bf Y}_I$, obtain $\hat\bomg$ via Glasso or Concord;
        \item Compute $\hat{{\bf A}} = {\bf I} - \hat{\bf D}\hat\bomg$;
        \item The predicted value is $\hat{{\bf Y}}_O$ whose $(i,j)$-th element is:
        \begin{equation*}
            (\hat{{\bf Y}}_O)_{i,j} = \sum_{k\neq i}\hat{{\bf A}}_{i,k}({\bf Y}_O)_{k,j}
        \end{equation*}
    \end{enumerate}
    
    \item Mixed model $Y = \rho {\bf A} Y + {\bf B} X + G$:
    \begin{enumerate}
        \item Obtain $\hat{{\bf A}}$ the same as in the graphical representation model;
        \item Obtain the least square estimate for $\rho$ by solving\footnote{
        The set of candidate values is $\mathcal{S} = \{\zeta\leq\kappa\leq\eta: {\bf I} - \kappa\hat{{\bf A}}\text{ is invertible.}\}$. In the empirical study, we choose $\zeta = -2, \eta = 4$ such that the optimal value can be covered by the interval.
        }:
        \begin{equation*}
            \hat\rho = \argmin_{\kappa\in\mathcal{S}}\quad\big\|({\bf I} - \kappa\hat{{\bf A}}){\bf Y}_I({\bf I} - {\bf X}_I^{\top}({\bf X}_I{\bf X}_I^{\top})^{-1}{\bf X}_I)\big\|_F^2
        \end{equation*}
        \item The predicted value is $\hat{{\bf Y}}_O$ whose $(i,j)$-th element is:
        \begin{equation*}
            (\hat{{\bf Y}}_O)_{i,j} = \hat\rho\sum_{k\neq i}\hat{{\bf A}}_{i,k}({\bf Y}_O)_{k,j} + \sum_{h = 1}^K\hat{{\bf B}}_{i,h}({\bf X}_O)_{h,j}
        \end{equation*}
    \end{enumerate}
\end{itemize}

\noindent For all the above models, the out-of-sample performance is measured using root mean squared error (RMSE) \cite{Hyndman2006}:

\begin{equation} \label{rmse}
    \text{RMSE} = \sqrt{\frac{1}{pn_O}\|\hat{{\bf Y}}_O - {\bf Y}_O\|_F^2}
\end{equation}

\noindent The relative RMSE (RMSE ratio), expressed as a percentage, is defined as
\begin{equation}\label{rmse-ratio}
    \text{RMSE}(\%) = \frac{\sqrt{\frac{1}{pn_O}\|\hat{\bf Y}_O - {\bf Y}_O\|_F^2}}{\sqrt{\frac{1}{pn_O}\|{\bf Y}_O\|_F^2}}
\end{equation}

\noindent The results for S\&P 500 and DJIA component stocks are shown in Table \ref{tab:rmse1} (More detailed results are shown in the Appendix). This table shows that GRM has a relatively lower RMSE than that of Fama-French and PCA factor models. In particular, the PCA factor model has the smallest RMSE for Dow Jones return data. However, its RMSE increases for S\&P 500 return data. According to \cite{Shen2016}, this is because as the sample sizes $n_I$ and $n_O$ are fixed while the value of $p$ gets larger, the deviation of the estimated eigenvectors from the true eigenvectors gets larger. As a result, the model's performance is deteriorating. In addition, by comparing the out-of-sample performance of the graphical representation model and the mixed model, we find out that in the mixed model the RMSE increases after factors are included. Such an increase of the RMSE indicates that adding factor terms actually does not provide additional benefit to characterize asset returns. The graphical term ${\bf A}$ is enough to capture the variation of asset returns.

\begin{table}[ht!]
    \centering
    \begin{tabular}{c|c|c||c}
\hline
\multicolumn{2}{c|}{}& Financial Crisis & Economic Expansion\\\hline
\multirow{3}{*}{In-sample}
& Start date & 2008-01-01 &2018-01-01\\
& End date & 2008-03-31 & 2018-03-31\\\cline{2-4}
& $n_I$ & 61 & 61 \\\hline
\multirow{3}{*}{Out-of-sample}
& Start date & 2008-04-01 & 2018-04-01\\ 
& End date & 2008-07-01 & 2018-07-01\\\cline{2-4}
& $n_O$ & 65  & 64\\\hline
\multirow{2}{*}{Value of $p$}
& Dow Jones & 28 & 30\\
& S\&P 500 & 447 & 498\\\hline
    \end{tabular}
    \caption{In the empirical study, two types of dataset are used: all available component stock returns of S\&P 500 and DJIA. For each type of dataset, the number of component stocks ($p$), the in-sample and out-of-sample sizes ($n_I$ and $n_O$), and the time periods of the in-sample data and the out-of-sample data are listed.}
    \label{tab:time}
\end{table}

\begin{table}[ht!]
    \centering
    \begin{tabular}{c|c|c|c||c|c}
    \hline\hline
    \multicolumn{6}{c}{\bf S\&P 500 Component Stocks}\\\hline
\multirow{2}{*}{Model} & \multirow{2}{*}{Model Type} & \multicolumn{2}{c||}{Financial Crisis} & \multicolumn{2}{c}{Economic Expansion}\\\cline{3-6}
                                    && RMSE & RMSE(\%) & RMSE & RMSE(\%)\\\hline
\multirow{1}{*}{FamaFrench} 
& 3 Factors                         & 1.98 & 85\% & 1.42 & 88\%\\\hline
\multirow{1}{*}{PCA}
& 3 Factors                         & 1.85 & 80\% & 1.38 & 86\%\\\hline
\multirow{2}{*}{GRM}
& Glasso $\hat\bomg$               & 1.73 & 75\% & 1.26 & 78\%\\
& Concord $\hat\bomg$              & 1.75 & 76\% & 1.27 & 79\%\\\hline
\multirow{2}{*}{Mixed}
& Glasso $\hat\bomg$ + 3 Factors   & 1.74 & 75\% & 1.27 & 79\%\\
& Concord $\hat\bomg$ + 3 Factors  & 1.76 & 76\% & 1.27 & 79\%\\\hline\hline
    \multicolumn{6}{c}{\bf DJIA Component Stocks}\\\hline
\multirow{2}{*}{Model} & \multirow{2}{*}{Model Type} & \multicolumn{2}{c||}{Financial Crisis} & \multicolumn{2}{c}{Economic Expansion}\\\cline{3-6}
                                    && RMSE & RMSE(\%) & RMSE & RMSE(\%)\\\hline
\multirow{1}{*}{FamaFrench} 
& 3 Factors                         & 2.28 & 86\% & 1.31 & 85\%\\\hline
\multirow{1}{*}{PCA}
& 3 Factors                         & 1.78 & 67\% & 1.19 & 78\%\\\hline
\multirow{2}{*}{GRM}
& Glasso $\hat\bomg$               & 2.17 & 82\% & 1.27 & 83\%\\
& Concord $\hat\bomg$              & 2.19 & 82\% & 1.28 & 83\%\\\hline
\multirow{2}{*}{Mixed}
& Glasso $\hat\bomg$ + 3 Factors   & 2.19 & 82\% & 1.30 & 85\%\\
& Concord $\hat\bomg$ + 3 Factors  & 2.23 & 84\% & 1.32 & 86\%\\\hline
    \end{tabular}
   \caption{The comparison of model performance measured by out-of-sample RMSE ($\times10^{-2}$) and the relative RMSE ($\%$), defined in (\ref{rmse}) and (\ref{rmse-ratio}). Results are produced using S\&P 500 and DJIA component stocks historical data from both the financial crisis period and the economic expansion period, defined in Table \ref{tab:time}. The lower the value of RMSE and relative RMSE, the better the model out-of-sample performance. More detailed results can be found in the Appendix of this paper.}
    \label{tab:rmse1}
\end{table}

\subsection{Properties of implied beta}
\label{implied-beta}

In this subsection, we focus on comparing the implied beta $\beta^{imp}_1$ from GRM and the market beta $\beta_1$ from the one-factor model. The market beta describes the movement of a security's returns responding to swings in the market and the concept of market beta is so prevailing in the financial market. For the notion simplicity, we use $\hat{\beta}^{imp}$ and $\hat{\beta}$ to denote the estimates for $\beta^{imp}_1$ and $\beta_1$, respectively. Our empirical study shows that the implied beta is similar to the beta provided by a factor model, as shown in Table \ref{tab:angle}. The estimated market beta, denoted as $\hat\beta$, is computed as the first column of $\hat{{\bf B}}$ from either the Fama-French factor model or the PCA factor model, normalized to have mean 1 (i.e., $\frac{1}{p}\hat\beta^{\top}e = 1$). The estimated implied beta, denoted as $\hat\beta^{imp}$, is computed as the first eigenvector of $\hat\bomg^{-1}$, normalized to have mean 1. The angle between the estimated market beta $\hat\beta$ and the estimated implied beta $\hat\beta^{imp}$ is:

\begin{equation}
    \text{Angle (in degree)} = \arccos\Big(\frac{\hat\beta^{\top}\hat\beta^{imp}}{\|\hat\beta\|_2\|\hat\beta^{imp}\|_2}\Big)\cdot\frac{180}{\pi}
\end{equation}

\begin{table}[ht!]
    \centering
    \begin{tabular}{c|cc|cc||cc|cc}
\hline\hline
\multicolumn{9}{c}{Angles (in degree) between $\hat\beta$ and $\hat{\beta}^{imp}$}\\\hline
\multirow{3}{*}{Models}
& \multicolumn{4}{c||}{\bf S\&P 500} & \multicolumn{4}{c}{\bf Dow Jones}\\\cline{2-9}
&\multicolumn{2}{c|}{\footnotesize Crisis} & \multicolumn{2}{c||}{\footnotesize Expansion} &\multicolumn{2}{c|}{\footnotesize Crisis} & \multicolumn{2}{c}{\footnotesize Expansion}\\\cline{2-9}
& {\footnotesize$\hat\beta(FF)$} & {\footnotesize$\hat\beta(PCA)$} & {\footnotesize$\hat\beta(FF)$} & {\footnotesize$\hat\beta(PCA)$} & {\footnotesize$\hat\beta(FF)$} & {\footnotesize$\hat\beta(PCA)$} & {\footnotesize$\hat\beta(FF)$} & {\footnotesize$\hat\beta(PCA)$} \\\hline
{\footnotesize$\hat\beta(PCA)$} & 5 & - & 1 & - & 15 & - & 3 & -\\
{\footnotesize$\hat{\beta}^{imp}(Glasso)$} & 4 & 4 & 3 & 3 & 10 & 18 & 6 & 7\\
{\footnotesize$\hat{\beta}^{imp}(Concord)$}\footnotemark & 4 & 5 & 6 & 6 & 11 & 18 & 7 & 7\\\hline
    \end{tabular}
    \caption{A comparison of angles (measured in degree of angle) of the market betas estimated from the Fama-French factor model (FF), the PCA factor model (PCA), and the GRM with graphical lasso (Glasso) or Concord (Concord) using the S\&P 500 and Dow Jones component stocks data from the financial crisis period and the economic expansion period. Angles between $\hat\beta$ and $\hat{\beta}^{imp}$ are small especially when the number of stocks becomes large. $\hat\beta$(FF) and $\hat\beta$(PCA) are the betas estimated from the Fama-French factor model and the PCA factor model, respectively. $\hat\beta^{imp}(Glasso)$ and $\hat\beta^{imp}(Concord)$ are implied beta estimated via graphical lasso and Concord, respectively.}
    \label{tab:angle}
\end{table}

\footnotetext{ In the empirical study, we found out that $\hat{\beta}^{imp}$ estimated using Concord contains extreme values, which may be the side-effect of the $\ell_1$ regularization on the concentration matrix. To remedy this, we posed an additional small Frobenius-norm penalty term in the optimization problem in producing the Concord results within 4.2. Since
\begin{equation*}
    \|\bomg\|_F^2 = \tr(\bomg^{\top}\bomg) = \tr\Big(({\bf U}\blam{\bf U}^{\top})^{\top}({\bf U}\blam{\bf U}^{\top})\Big) = \tr({\bf U}\blam^2{\bf U}^{\top}) = \tr(\blam^2) = \sum_{i=1}^p\lambda_i^2
\end{equation*}
\noindent where ${\bf U}\blam{\bf U}^{\top}$ is the spectral decomposition of $\bomg$ and $\lambda_i$ is the $i$-th eigenvalue, thus such a penalty shrinks the eigenvalues and avoids extreme values in the eigenvectors. See more details in \cite{Ali2017}.}

\noindent From Table \ref{tab:angle}, we see that $\hat\beta^{imp}$ has very similar angle results with $\hat\beta$, thus $\hat\beta^{imp}$ is similar to $\hat\beta$ and preserves the properties that $\hat\beta$ has. Besides, we have two additional observations: First, angles from S\&P 500 stocks are relatively smaller than those from Dow Jones stocks. This can be explained by the asymptotic result that as $p$ gets larger, $\beta^{imp}$ and $\beta$ converge to each other, thus the angles among the estimated betas become smaller. Second, for Dow Jones, betas from a crisis period have larger angles than those from an economic expansion period, and all of $\hat\beta(FF), \hat\beta(PCA), \hat\beta^{imp}$ have large angles among each other. This is likely due to the empirical evidence that during the crisis period, the estimated betas become more disperse than those during the expansion period. To obtain a further intuition into these differences, we define $d_v(\hat\beta)$ as a dispersion measurement of $\hat\beta = (\hat\beta_1, \hat\beta_2, \cdots, \hat\beta_p)^{\top}$, where:

\begin{equation}
\label{eq:dispersion}
    d_v(\hat\beta) = \sqrt{\frac{1}{p}\sum_{i=1}^p\Big(\frac{\hat\beta_i}{\bar\beta} - 1\Big)^2}\quad\text{where}\quad
    \bar\beta = \frac{1}{p}\sum_{i=1}^p\hat\beta_i
\end{equation}

\noindent The dispersion measures the beta dispersion around 1, which is the beta of the market portfolio. Empirically, the beta has a larger dispersion during the financial crisis than during the period of economic expansion. As is shown in Table \ref{tab:dispersion}, beta dispersion from a crisis period is larger than that from an expansion period. In particular, the Dow Jones component stocks in the crisis period have a higher variability of dispersion. In this case, the highest dispersion comes from the PCA factor model while the lowest dispersion comes from the Fama-French factor model, and the dispersions coming from the GRM are neither too high nor too low. Such a result raises the possibility that GRM provides a useful estimate for a market beta dispersion.

\begin{table}[ht!]
    \centering
    \begin{tabular}{c|c|c||c|c}
\hline\hline
\multicolumn{5}{c}{Dispersion $d_v(\cdot)$ of estimated betas}\\\hline
\multirow{2}{*}{Dispersion}
& \multicolumn{2}{c||}{\bf S\&P 500} & \multicolumn{2}{c}{\bf Dow Jones}\\\cline{2-5}
& Crisis & Expansion & Crisis & Expansion\\\hline
$d_v(\hat\beta(FF))$ & 0.431 & 0.333 & 0.382 & 0.331\\
$d_v(\hat\beta(PCA))$& 0.476 & 0.337 & 0.574 & 0.358\\
$d_v(\hat\beta^{imp}(Glasso))$& 0.439 & 0.297 & 0.410 & 0.362\\
$d_v(\hat\beta^{imp}(Concord))$& 0.418 & 0.265 & 0.425 & 0.328\\\hline
    \end{tabular}
    \caption{The dispersions, defined in $\req{eq:dispersion}$, of the market betas estimated from four models: the Fama-French factor model, the PCA factor model, the GRM with Glasso, and the GRM with Concord. All models are calibrated using the S\&P 500 and Dow Jones component stock historical data from the financial crisis period and the economic expansion period. The dispersions of the implied betas of the GRM are similar to the dispersions of the market betas estimated from other models. Thus our results indicate that the implied beta preserves the properties of the market beta coming from a factor model.}
    \label{tab:dispersion}
\end{table}

In addition, we explore the recovered market volatility of the graphical representation model, which is highly similar to the market volatility from the factor models. Annualized market volatility for each model is estimated as:

\begin{equation} \label{eq:volatility}
    \begin{split}
        \text{Fama-French factor model}:&\quad \sqrt{\var(({\bf X}_I)_{1\cdot})\cdot252}\times 100\%\\
        \text{PCA factor model}:&\quad \sqrt{\var\Big(\frac{1}{\hat\beta^{\top}\hat\beta}\hat\beta^{\top}\hat{{\bf Y}}_I\Big)\cdot252}\times 100\%\\
        \text{GRM}:&\quad \sqrt{\var\Big(\frac{1}{(\hat\beta^{imp})^{\top}\hat\beta^{imp}}(\hat\beta^{imp})^{\top}\hat{{\bf Y}}_I\Big)\cdot252}\times 100\%
    \end{split}
\end{equation}

\noindent where $({\bf X}_I)_{1\cdot}$ represents the 1st row of ${\bf X}_I$ and 252 is the number of trading days in a year. Table \ref{tab:vol} shows that the annualized market volatility recovered from the graphical representation model is similar to that from a factor model.

\begin{table}[ht!]
    \centering
    \begin{tabular}{c|cc||cc}
\hline\hline
\multicolumn{5}{c}{Annualized volatility of market returns}\\\hline
\multirow{2}{*}{Model}
& \multicolumn{2}{c||}{\bf S\&P 500} & \multicolumn{2}{c}{\bf Dow Jones}\\\cline{2-5}
& Crisis & Expansion & Crisis & Expansion\\\hline
FamaFrench factor model & 25.98\% & 18.08\% & 23.93\% & 18.75\%\\
PCA factor model        & 25.61\% & 18.06\% & 22.71\% & 18.62\%\\
GRM (Glasso)  & 25.92\% & 18.24\% & 23.58\% & 18.48\%\\
GRM (Concord) & 26.07\% & 18.33\% & 23.41\% & 18.65\%\\\hline
    \end{tabular}
    \caption{The recovered annualized market volatility is defined in $\req{eq:volatility}$ and is estimated from four models: the Fama-French factor model, the PCA factor model, the GRM with Glasso, and the GRM with Concord. The models are calibrated using the S\&P 500 and Dow Jones component stock historical data from both the financial crisis period and the economic expansion period. The result shows that the recovered annualized market volatilities from GRM are similar to those from other models, thus implied beta from the GRM preserves the properties of a market beta.}
    \label{tab:vol}
\end{table}

\subsection{Graph visualization}

In this subsection, we explore the graph visualization in the graphical representation model. The graph visualization has a meaningful interpretation for investment decision making. It allows an investor to identify sectors or communities of stocks that are positively or negatively partially correlated, thus provides an intuition for equities selection in portfolio construction.  In the graphical representation model, associations among asset returns can be fully visualized by an undirected graph or network. One may be interested in the network of all assets that is constructed based on the estimated partial correlation. Given the estimated precision matrix $\hat{\bomg}$ that is estimated either from Glasso or Concord, the estimated partial correlation matrix $\hat{{\bf P}} = ((\hat\varrho_{ij}))$ can be given by \cite{Baba2004}:

\begin{equation}
\label{eq:partial correlation matrix}
    \hat{{\bf P}} = -\hat{\bf D}^{1/2}\hat\bomg\hat{\bf D}^{1/2}
    \quad \text{with}\quad \hat{\bf D} = \diag(\hat\bomg)^{-1}
\end{equation}

\noindent We plot the network that is characterized by $\hat{\bf P}$, where a blue edge between the $i$-th  and  the $j$-th  node  represents $\hat\varrho_{ij} > 0$,  and  a  red  edge  represents $\hat\varrho_{ij} < 0$. The full graphs and related details can be found in the Appendix of this paper.

Based on the sparsity patterns of $\hat{\bf P}$ (i.e., the zero/non-zero locations of $\hat{\bf P}$), we conduct community detection using the random walk method proposed by \cite{Pons2005}. The random walk method tries to find densely connected subgraphs via detecting short random walks. We use the \texttt{cluster\_walktrap} function from the \textbf{R} package \texttt{igraph}\footnote{The package and its reference manual can be downloaded from \url{https://cran.r-project.org/web/packages/igraph/index.html}.} We consider the sector category\footnote{We use Global Industry Classification Standard (GICS) sector category. It is an industry taxonomy developed in 1999 by MSCI and Standard \& Poor's and consists of 11 sectors. See \url{https://en.wikipedia.org/wiki/Global_Industry_Classification_Standard} for more details and see \url{https://en.wikipedia.org/wiki/List_of_S\%26P_500_companies} for the sector information of S\&P 500 stocks.} and find out that many communities are dominated by certain sectors, while stocks from some sectors are conglomerated into one community. Interestingly, we discover that stocks coming from different sectors while classified into one community have highly similar business operations. Such a similarity is reflected in the co-movement of their stock prices thus in the data-driven community detection procedure. For example, we use the S\&P 500 data from the expansion period and use Glasso to estimate the precision matrix $\hat\bomg$, and conduct the community detection based on $\hat\bomg$. A sub-graph of Community 4 is shown in Figure \ref{fig:cmm4}.

\begin{figure}[H]
    \centering
    \includegraphics[scale = 0.54]{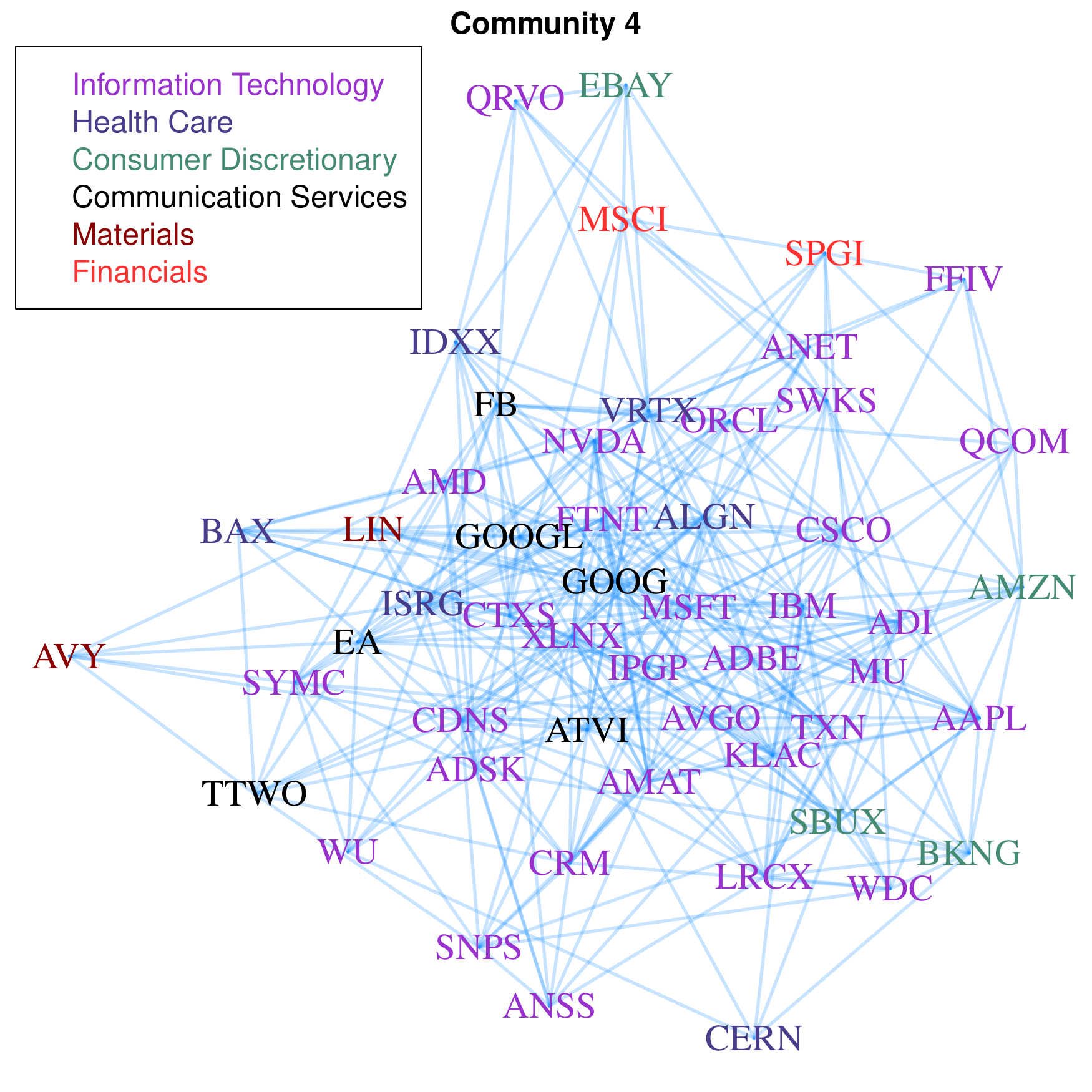}
    \caption{The recovered graph for one community (Community 4) detected from the sparse estimation $\hat\bomg$, with which the community detection is conducted. In the community detection, 11 communities are chosen to be detected (same number of sectors in GICS). The sparse $\hat\bomg$ is estimated using the S\&P 500 Index component stock data from the economic expansion period defined in Table \ref{tab:time}. In this community, the dominant sector is the information technology. Despite there are companies from other sectors, the business of those companies are, to some extend, related to the information technology.}
    \label{fig:cmm4}
\end{figure}

In addition, stocks in a certain sector are split into different communities. For example, we find out that the Communication Services (CS) sector has 9 companies in Community 3 and 6 companies in Community 4. The sub-graphs and a comparison of these 15 companies are shown in Figure \ref{fig:communication}. It seems that within each community, the types of businesses of these companies are highly similar while compared between communities the types of business become more distinct. Such a comparison implies that GICS Sector category follows an industry classification benchmark and is not a data-driven approach to categorize all stocks. Since the community-based category is purely data-driven thus provides a data-driven way to label all stocks. From this aspect, the communities that are detected from the GRM graph provide better insight into investment decision making and portfolio construction.

\begin{figure}[H]
\centering
\begin{subfigure}{0.48\textwidth}
\includegraphics[width=\textwidth]{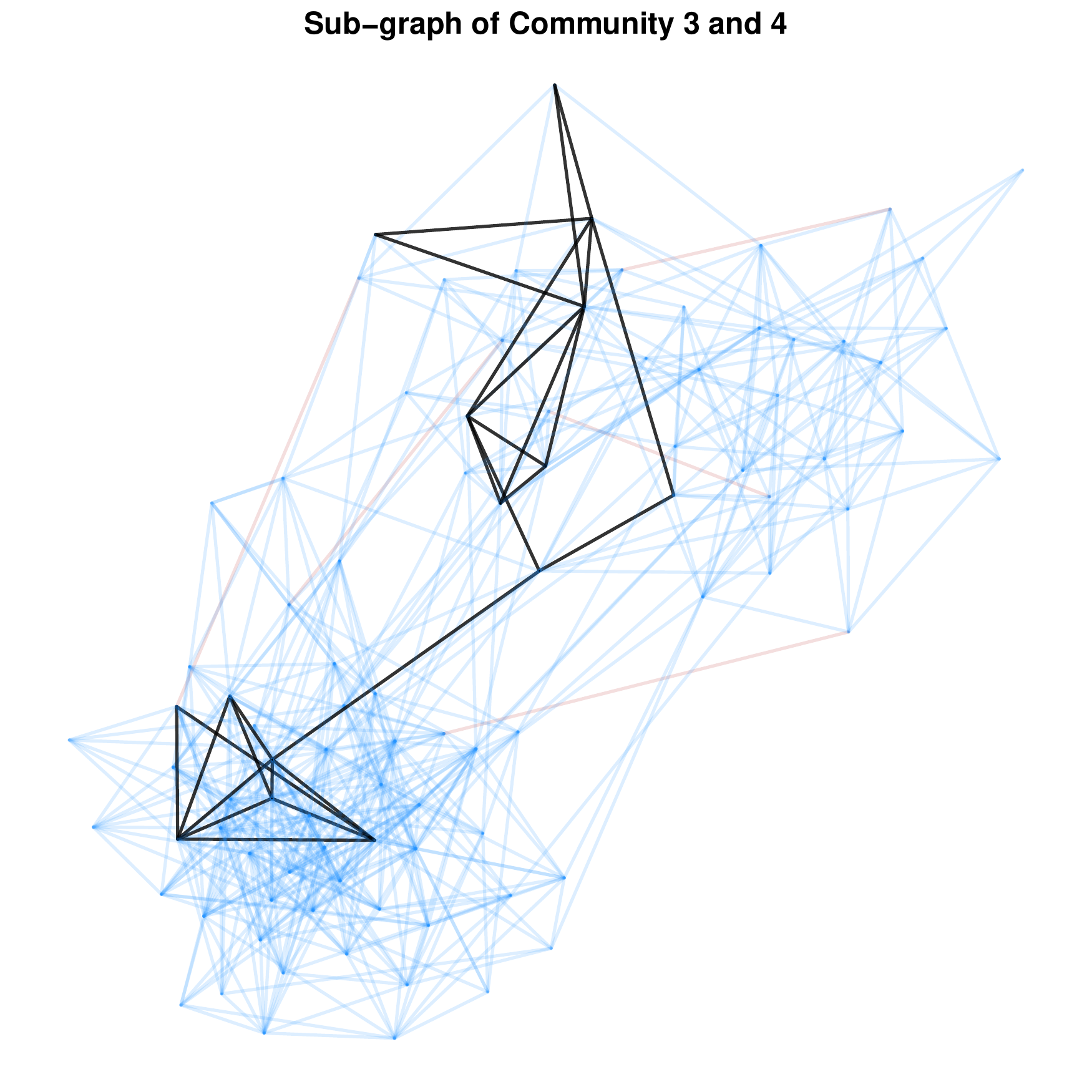}
\end{subfigure}
\begin{subfigure}{0.48\textwidth}
\includegraphics[width=\textwidth]{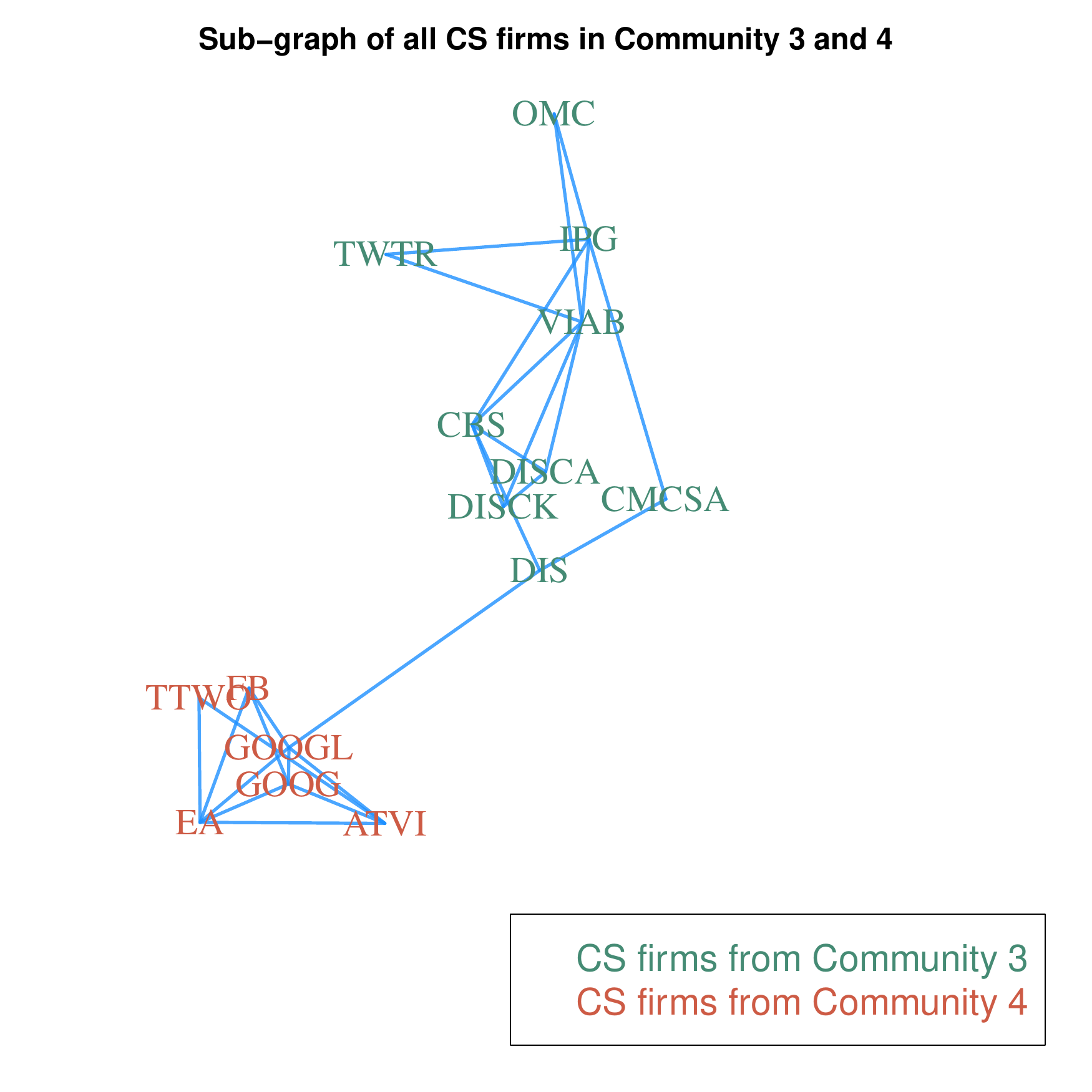}
\end{subfigure}
\caption{
 Recovered graph (above) of two communities (Community 3 and 4) detected from the sparse graphical model using S\&P 500 component stock data from the economic expansion period, and the corresponding graph of all communication service (CS) companies in the two communities (below). CS firms in one community are mostly movies related companies while in another community are broadly online and gaming related companies. Full names of CS firms are listed in Table \ref{tab:aa1}.}
\label{fig:communication}
\end{figure}

\begin{table}[ht!]
    \centering
\begin{tabular}{rl|rl|rl}
    \hline
    \multicolumn{6}{c}{CS companies in Community 3}  \\\hline
    {\bf CBS}: & CBS Corp.                &  {\bf CMCSA}: & Comcast Corp.          & {\bf DIS}: & Disney Company.\\
    {\bf DISCA}: & Discovery, class A.&  {\bf DISCK}:& Discovery, class C. &  {\bf IPG}: & Interpublic Group.\\
    {\bf OMC}: & Omnicom Group.           &  {\bf TWTR}: & Twitter, Inc.           & {\bf VIAB}: & Viacom Inc. \\\hline
    \multicolumn{6}{c}{CS companies in Community 4}  \\\hline
    {\bf ATVI}: & Activision Blizzard.    &  {\bf EA}: & Electronic Arts.         & {\bf FB}: & Facebook, Inc. \\
    {\bf GOOG}: & Alphabet, class C.  & {\bf GOOGL}: & Alphabet, class A. &  {\bf TTWO}: & Take-Two Interactive.\\\hline
    \end{tabular}
    \caption{Stock symbols and corresponding companies of communication service sector in Community 3 and Community 4.}
    \label{tab:aa1}
\end{table}

\section{Summary and future work}
\label{summary}

This paper proposes the endogenous representation model that provides an endogenous perspective on the observed variance of equity markets. Compared with traditional factor models such as Fama-French or PCA factor models, GRM offers a different perspective on the interpretation of asset returns and the inter-dependencies among them, and relates the interpretation to the popularly adopted graphical models, which allow to visualize the dependency structure among asset returns. In addition to the different perspective from GRM, empirical evidence also shows that GRM has similar out-of-sample model performance with a factor model, and preserves desirable properties of the market beta from a factor model. 

One potential future work could be to study the relationship between the Markov diffusion process and the graphical model. On the one hand, the graphical model in the Gaussian setup is a Markov random field that essentially corresponds to a diffusion process \cite{Kadanoff2001}. On the other hand, the adjacency matrix of the term ${\bf A}$ defines a Laplacian matrix \cite{Bapat2014}, which is the generator of a Markov diffusion process \cite{Chung1997}. In our setting, such a diffusion process is modeling the spread of the price information of stocks and contributes to the price formation (i.e., the variations of asset returns). What is the relationship between the implied market beta and the resolvent of a Markov process? Such a question, together with the above argument, could be one interesting research topic that extends the work of this paper.

\appendix
\section{Appendices}
\label{chapter2-appendix}

\subsection{Proof of Lemma \ref{lem:uncorrelated}}

\begin{proof}
Without loss of generality, we assume $\mathbb{E}(Y)=\mathbf{0}$. Given $\cov(Y_j, E_i) = 0$, we have $\mathbb{E}(E_iY_{-i}^{\top}) = 0$ where $Y_{-i} = (Y_1, \cdots, Y_{i-1}, Y_{i+1}, \cdots, Y_p)^{\top}$ for $i = 1, 2, \cdots, p$, thus $\mathbb{E}(EY^{\top}) = {\bf D}$ for some diagonal matrix ${\bf D}$. From (\ref{eq:endogenous}),
\begin{equation}
    \mathbb{E}(YY^{\top}) = \mathbb{E}({\bf A} YY^{\top}) + \mathbb{E}(EY^{\top})
\end{equation}
Thus
\begin{equation}
    \bsig = {\bf A}\bsig + {\bf D}
\end{equation}
So 
\begin{equation}
    ({\bf I} - {\bf A}) = {\bf D}\bsig^{-1} = {\bf D}\bomg
\end{equation}
Since $a_{ii} = {\bf A}_{ii} = 1-{\bf D}_{ii}\bomg = 0$, ${\bf D}_{ii} = 1/\omega_{ii}$, thus we have the results (\ref{eq:results}).
\end{proof}

\subsection{Proof of Theorem \ref{lem:mini}}

\begin{proof}
Without loss of generality, we assume $\mathbb{E}(Y) = \mathbf{0}$. Let ${\bf M} = (({\bf M}_{ij}))\in\scrA$ and $Y_i$ be the $i$-th element of $Y$. Denote the $i$-th row of ${\bf M}$ as ${\bf M}_{i\cdot}$. Also, for an arbitrary $p\times p$ matrix ${\bf X} = (({\bf X}_{ij}))$, denote ${\bf X}_{i, -i} = ({\bf X}_{i1}, {\bf X}_{i2}, \cdots, {\bf X}_{i,i-1}, {\bf X}_{i, i+1}, \cdots, {\bf X}_{ip})$, i.e., ${\bf X}_{i,-i}$ is the $i$-th row of ${\bf X}$ without the $i$-th element ${\bf X}_{ii}$, and ${\bf X}_{-i, i} = {\bf X}_{i,-i}^{\top}$, and ${\bf X}_{-i, -i}$ be the matrix after the $i$-th column and the $i$-th row of ${\bf X}$ been removed. We have

\begin{equation}
    \mathbb{E}(|Y - {\bf M} Y|^2) = \mathbb{E}\Big(\sum_{i=1}^p(Y_i-{\bf M}_{i\cdot}Y)^2\Big)
    = \sum_{i=1}^p\mathbb{E}[(Y_i - {\bf M}_{i\cdot}Y)^2]
\end{equation}

\noindent Since ${\bf M}_{i\cdot}$ and ${\bf M}_{j\cdot}$ contain no common variables for any $i\neq j$, to minimize $\req{eq:resvar}$ is equivalent to minimize each $\mathbb{E}[(Y_i - {\bf M}_{i\cdot}Y)^2]$ for $i = 1, 2, \cdots, p$. For a fixed $i$, let 

\begin{equation}
    \tilde{\bsig} = \begin{pmatrix}
    \bsig_{ii} & \bsig_{i,-i}\\
    \bsig_{-i,i} & \bsig_{-i,-i}
    \end{pmatrix}
    \quad\text{and}\quad
    \tilde{\bomg} = \begin{pmatrix}
    \bomg_{ii} & \bomg_{i,-i}\\
    \bomg_{-i,i} & \bomg_{-i,-i}
    \end{pmatrix}
\end{equation}

\noindent Since $\bomg\bsig = {\bf I}$, it is easy to show that $\tilde{\bomg}\tilde{\bsig} = {\bf I}$, thus
\begin{equation}
\label{use1123}
    \bomg_{ii}\bsig_{i,-i} + \bomg_{i,-i}\bsig_{-i,-i} = \mathbf{0}
\end{equation}
\noindent We have
\begin{equation}
\begin{split}
\mathbb{E}[(Y_i - {\bf M}_{i\cdot}Y)^2] &= \mathbb{E}\Big[(Y_i - \sum_{j\neq i}{\bf M}_{ij}Y_j)^2\Big]\\
 &= \mathbb{E}(Y_i^2) - 2\mathbb{E}\Big(\sum_{j\neq i}{\bf M}_{ij}Y_iY_j\Big) + \mathbb{E}[\Big(\sum_{j\neq i}{\bf M}_{ij}Y_j\Big)^2]\\
 &= \bsig_{ii} - 2{\bf M}_{i,-i}\bsig_{-i,i} +{\bf M}_{i,-i}\bsig_{-i,-i}{\bf M}_{i,-i}^{\top} \triangleq f({\bf M}_{i,-i}) 
 \end{split}
 \end{equation}
 \noindent Consider the first-order condition
 
 \begin{equation}
 \label{firstordercondition}
 \frac{df({\bf M}_{i,-i})}{d{\bf M}_{i,-i}} = -2\bsig_{-i,i} + 2\bsig_{-i,-i}{\bf M}_{i,-i}^{\top} = \mathbf{0}
\end{equation}

\noindent Solve $\req{firstordercondition}$ and follow the result $\req{use1123}$ we have

\begin{equation}
    {\bf M}_{i,-i} = \bsig_{i,-i}\bsig_{-i,-i}^{-1} = -\bomg_{ii}^{-1}\bomg_{i,-i}
\end{equation}

\noindent Since ${\bf M}\in\scrA$, ${\bf M}_{ii} = 0$ and 

\begin{equation}
    {\bf M}_{i,\cdot} = e_i -\bomg_{ii}^{-1}\bomg_{i,\cdot}
\end{equation}
\noindent where $e_i$ is a $p$ dimensional row vector with only the $i$-th element is 1 and all the remaining elements are 0. Thus

\begin{equation}
    {\bf M} = \begin{pmatrix}
    {\bf M}_{1\cdot}\\
    {\bf M}_{2\cdot}\\
    \cdots\\
    {\bf M}_{p\cdot}
    \end{pmatrix}
    = {\bf I} - {\bf D}\bomg \triangleq {\bf A}
    \quad\text{with}\quad
    {\bf D} = \diag(\bomg)^{-1}
\end{equation}

\end{proof}

\subsection{Proof of Theorem \ref{thm:partition}}

\begin{proof}
Given $Y_{\mathbb{I}|\mathbb{J}} \triangleq Y_{\mathbb{I}} - \mathcal{L} (Y_{\mathbb{J}})$, where $\mathcal{L} (Y_{\mathbb{J}}) = \bsig_{\mathbb{IJ}} \bsig_{\mathbb{JJ}}^{-1} Y_{\mathbb{J}}$ and 
\begin{equation}
    \bomg = \bsig^{-1} = \begin{pmatrix}
    \bsig_{\mathbb{I}\mathbb{I}} & \bsig_{\mathbb{I}\mathbb{J}}\\
    \bsig_{\mathbb{J}\mathbb{I}} & \bsig_{\mathbb{J}\mathbb{J}}
    \end{pmatrix}^{-1}
    =\begin{pmatrix}
    \bomg_{\mathbb{I}\mathbb{I}} & \bomg_{\mathbb{I}\mathbb{J}}\\
    \bomg_{\mathbb{J}\mathbb{I}} & \bomg_{\mathbb{J}\mathbb{J}}
    \end{pmatrix}
\end{equation}
It is easy to obtain $\bomg_{\mathbb{I}\mathbb{I}}^{-1} = \bsig_{\mathbb{I}\mathbb{I}} - \bsig_{\mathbb{I}\mathbb{J}}\bsig_{\mathbb{J}\mathbb{J}}^{-1}\bsig_{\mathbb{J}\mathbb{I}}$ and the variance-covariance matrix of the random vector $Y_{\mathbb{I}|\mathbb{J}}$ as
\begin{equation}
\begin{split}
    \Var(Y_{\mathbb{I}|\mathbb{J}}) &= \Cov(Y_{\mathbb{I}} - \mathcal{L} (Y_{\mathbb{J}}),Y_{\mathbb{I}} - \mathcal{L} (Y_{\mathbb{J}}))\\
    &= \bsig_{\mathbb{I}\mathbb{I}} - \bsig_{\mathbb{I}\mathbb{J}}\bsig_{\mathbb{J}\mathbb{J}}^{-1}\bsig_{\mathbb{J}\mathbb{I}}=\bomg_{\mathbb{I}\mathbb{I}}^{-1}
    \end{split}
\end{equation}
Right multiply $Y_{\mathbb{I}|\mathbb{J}}^{\top}$ and take expectation on both sides of Equation (\ref{eq:38}), we have $\bomg_{\mathbb{I}\mathbb{I}}^{-1} = {\bf A}_{\mathbb{I}|\mathbb{J}}\bomg_{\mathbb{I}\mathbb{I}}^{-1} + {\bf D}_{\mathbb{I}|\mathbb{J}}$ where ${\bf D}_{\mathbb{I}|\mathbb{J}} = \mathbb{E}(E_{\mathbb{I}|\mathbb{J}}Y_{\mathbb{I}|\mathbb{J}}^{\top})$. Thus
\begin{equation}
    {\bf I} - {\bf A}_{\mathbb{I}|\mathbb{J}} = {\bf D}_{\mathbb{I}|\mathbb{J}}\bomg_{\mathbb{I}\mathbb{I}}
\end{equation}
Given $({\bf A}_{\mathbb{I}|\mathbb{J}})_{ii} = 0$ and Equation (\ref{eq:39}) we have ${\bf D}_{\mathbb{I}|\mathbb{J}}$ is diagonal and $({\bf D}_{\mathbb{I}|\mathbb{J}})_{ii} = \frac{1}{(\bomg_{\mathbb{I}\mathbb{I}})_{ii}}$, thus ${\bf D}_{\mathbb{I}|\mathbb{J}} = \diag(\bomg_{\mathbb{I}\mathbb{I}})^{-1}$ and
\begin{equation}
\begin{split}
    {\bf A}_{\mathbb{I}|\mathbb{J}} &= {\bf I} - {\bf D}_{\mathbb{I}|\mathbb{J}}\bomg_{\mathbb{I}\mathbb{I}}
    = {\bf I} - \diag(\bomg_{\mathbb{I}\mathbb{I}})^{-1}\bomg_{\mathbb{I}\mathbb{I}}\\
    &= ({\bf I} - \diag(\bomg)^{-1}\bomg)_{\mathbb{I}\mathbb{I}}
    = {\bf A}_{\mathbb{I}\mathbb{I}}
    \end{split}
\end{equation}
where ${\bf A} = {\bf I} - {\bf D}\bomg, {\bf D}=\diag(\bomg)^{-1}$.
\end{proof}

\subsection{Proportion of variance explained for PCA-factor model and GRM}

This study is to examine the out-of-sample performance for a factor model and the GRM in terms of how much variation is explained over a longer period of time. During the time horizon under study, both models are re-calibrated periodically. Specifically, we compare the out-of-sample performance between the PCA 5-factor model and the GRM over the S\&P 500 component stocks. The chosen time horizon ranges from 2000 to 2018, during which both models are re-calibrated every three months. At each re-calibration time point, the model is re-calibrated using the historical asset returns within three months before the time point. Then the calibrated model is applied to the next three-month period going forward. The out-of-sample performance is measured by the $R^2$ score (the coefficient of determinant). $R^2$ score represents the proportion of variance of asset returns that has been explained by the models. In our study, we use the definition of $R^2$ to evaluate the performance for each asset $i$ as follows \footnotetext{See \url{https://scikit-learn.org/stable/modules/model_evaluation.html\#r2-score}}: 
\begin{equation}
    R^2({\bf Y}_O^{(i)}, \hat{\bf Y}_O^{(i)}) = 1 - \frac{\sum_{j=1}^{n_O}(({\bf Y}_O)_{i,j}-(\hat{\bf Y}_O)_{i,j})^2}{\sum_{j=1}^{n_O}(({\bf Y}_O)_{i,j} - \bar Y_O^{(i)})^2}
\end{equation}
where $\bar Y_O^{(i)}$ denotes the sample mean of the $i$-th row of ${\bf Y}_O$. We evaluate the proportion of variance explained of the models by computing the mean score of $R^2({\bf Y}_O^{(i)}, \hat{\bf Y}_O^{(i)})$, i.e.:
\begin{equation}
\label{eq:rs}
    R^2 = \frac{1}{p}\sum_{i=1}^pR^2({\bf Y}_O^{(i)}, \hat{\bf Y}_O^{(i)})
\end{equation}
\noindent Figure \ref{fig:pve} illustrates the comparison of the proportions of variance explained for PCA 5-factor model and GRM (using either Glasso or Concord). The comparison indicates that, empirically, GRM model preserves a similar proportion of variance explained with that of a factor model.

\subsection{Detailed out-of-sample performance} \label{oos}

We denote $\kappa$ as the number of free parameters in a model. We specify $\kappa$ as follows:

\begin{enumerate}
    \item {\bf Fama-French $k$-factor model}: Parameters of a Fama-French $k$-factor model includes the elements in the matrix of factor exposure ${\bf B}$, thus $\kappa = pk$.
    
    \item {\bf PCA $k$-factor model}: Parameters include elements of the matrix for ${\bf B}$, thus $\kappa = pk$.
    
    \item {\bf Mixed model with spatial interactions and $k$ factors}: Parameters include the coefficient $\rho$ for the spatial interactions and the factor exposures in the matrix ${\bf B}$, thus $\kappa = 1 + pk$.
    
    \item {\bf Mixed model with graphical interactions and $k$ factors}: Since coefficients of graphical interactions are estimated via the precision matrix, parameters include the coefficient $\rho$ and the factor exposures in ${\bf B}$. Once the precision matrix and the factor exposures are given, there is no free parameters in the covariance matrix of the residual $G$, thus $\kappa = 1 + pk + \frac{1}{2}p(p+1) - \frac{1}{2}g(\hat\bomg)$ where $g(\hat\bomg)$ denotes the total number of zero elements in $\hat\bomg$.
    
    \item {\bf GRM}: Since GRM is fully characterized by the precision matrix $\bomg$, we have $\kappa = \frac{1}{2}p(p+1) - \frac{1}{2}g(\hat\bomg)$.
\end{enumerate}

Following the above definition of $\kappa$, we use the out-of-sample root mean squared error (RMSE) and Bayesian Information Criterion (BIC, proposed by \cite{Schwarz1978}) as two additional measurements for model performance. The BIC proposed by \cite{Peng2009} is:

\begin{equation} \label{eq:bic}
    BIC = n_O\sum_{i = 1}^p\log(RSS_i) + \kappa\times\log(n_O)
\end{equation}

\noindent where

\begin{equation}
    RSS_i = \frac{1}{n_O}\sum_{j = 1}^{n_O}\Big(({\bf Y}_{O})_{i,j} - (\hat{{\bf Y}}_O)_{i,j}\Big)^2
\end{equation}

All out-of-sample performance results are shown in Table \ref{tab:performance}. The lower the values of BIC and RMSE, the more robust the model would be. The out-of-sample RMSE indicates that a sparse ${\bf A}$ is sufficient to model asset returns and adding factors is not beneficial. The out-of-sample BIC implies that despite the better goodness-of-fit of GRM, under our dataset, the model requires more number of parameters.

\begin{table}[ht!]
    \centering
    \begin{tabular}{c|c|c|c|c|c||c|c|c|c}
    \hline\hline
    \multicolumn{9}{c}{\bf S\&P 500 Component Stocks}\\\hline
\multirow{2}{*}{Model} & \multirow{2}{*}{Model Type} & \multicolumn{4}{c||}{Financial Crisis} & \multicolumn{4}{c}{Economic Expansion}\\\cline{3-10}
                                    && No.P   & BIC     & {\footnotesize RMSE} & {\footnotesize RMSE(\%)} & No.P  & BIC & {\footnotesize RMSE} & {\footnotesize RMSE(\%)}\\\hline
\multirow{2}{*}{\footnotesize FamaFrench} 
& 3 Factors                         &  1341   & -23.28   & 1.98 & 85\% & 1494 & -27.50 & 1.42 & 88\%\\
& 5 Factors                         &  2235   & -22.92   & 1.96 & 85\% & 2490 &-27.10  & 1.42 & 88\%\\\hline
\multirow{2}{*}{PCA}
& 3 Factors                         &  1341   & -23.61   & 1.85 & 80\% & 1494 &-27.75 & 1.38 & 86\%\\
& 5 Factors                         &  2235   & -23.32   & 1.80 & 78\% & 2490 &-27.48 & 1.34 & 83\%\\\hline
\multirow{3}{*}{Mixed}
& {\footnotesize Spatial ${\bf W}$ + 3 Factors} & 1342     & -23.31 & 1.96 & 85\% & 1495 & -27.51 & 1.41 & 87\%\\
& {\footnotesize Glasso $\hat{{\bf A}}$ + 3 Factors} & 8065     & -21.22 & 1.74 & 75\% & 9191 & -25.42 & 1.27 & 79\%\\
& {\footnotesize Concord $\hat{{\bf A}}$ + 3 Factors} & 5101     & -22.48 & 1.76 & 76\% & 6373 & -26.62 & 1.27 & 79\%\\\hline
\multirow{2}{*}{GRM}
& Glasso $\hat{{\bf A}}$               & 6723      & -21.81 & 1.73 & 75\% & 7696 & -26.10  & 1.26 & 78\%\\
& Concord $\hat{{\bf A}}$              & 3759      & -22.98 & 1.75 & 76\% & 4878 & -27.27  & 1.27 & 79\%\\\hline

    \hline\hline
    \multicolumn{9}{c}{\bf Dow Jones Component Stocks}\\\hline
\multirow{2}{*}{Model} & \multirow{2}{*}{Model Type} & \multicolumn{4}{c||}{Financial Crisis} & \multicolumn{4}{c}{Economic Expansion}\\\cline{3-10}
                                && No.P & BIC    & {\footnotesize RMSE} & {\footnotesize RMSE(\%)} & No.P  & BIC & {\footnotesize RMSE} & {\footnotesize RMSE(\%)}\\\hline
\multirow{2}{*}{\footnotesize FamaFrench}
& 3 Factors                     & 84   & -1.46 &  2.28 & 86\%  & 90 & -1.67 & 1.31 & 85\%\\
& 5 Factors                     & 140  & -1.44 &  2.22 & 84\% & 150 & -1.64 & 1.42 & 93\%\\\hline
\multirow{2}{*}{PCA}
& 3 Factors                     & 84  & -1.49 &  1.78 & 67\% & 90 & -1.69 & 1.19 & 78\%\\
& 5 Factors                     & 140 & -1.50 &  1.60 & 60\% & 150 & -1.71 & 1.09 & 71\%\\\hline
\multirow{3}{*}{Mixed}
& {\footnotesize Spatial ${\bf W}$ + 3 Factors} & 85    &-1.46 & 2.27 & 85\% & 91  & -1.67 & 1.31 & 85\%\\
& {\footnotesize Glasso $\hat{{\bf A}}$ + 3 Factors} & 272 &-1.37 & 2.19 & 82\% & 316 & -1.58 & 1.30 & 85\%\\
& {\footnotesize Concord $\hat{{\bf A}}$ + 3 Factors} & 235& -1.38 & 2.23 & 84\% & 267 & -1.59 & 1.32 & 86\%\\\hline
\multirow{2}{*}{GRM}
& Glasso $\hat{{\bf A}}$           & 187   &-1.42 & 2.17 & 82\% & 225 & -1.63 & 1.27 & 83\%\\
& Concord $\hat{{\bf A}}$          & 150   &-1.42 & 2.19 & 82\% & 176 & -1.64 & 1.28 & 83\%\\\hline
    \end{tabular}
    \caption{The out-of-sample RMSE ($\times 10^{-2}$), BIC ($\times 10^4$, defined in $\req{eq:bic}$), and the ratios of RMSE over the total returns variation (RMSE(\%), defined in \req{rmse-ratio}) of four types of model. All empirical results are produced using Dow Jones and S\&P 500 component stock historical data from both the financial crisis period and the economic expansion period. No.P denotes the number of free parameters. The lower the values of BIC, RMSE, and the RMSE(\%), the better the model out-of-sample performance.}
    \label{tab:performance}
\end{table}

\subsection{The implications of implied beta}

The implied beta from the GRM has two implications:

\begin{enumerate}
    \item \textit{Elements of an implied beta are almost positive}. Given the lemma\footnote{See Lemma 3.5.2 from Daniel A. Spielman \textit{Spectral Graph Theory} lecture notes, Lecture 3, which is downloadable from \url{https://www.cs.yale.edu/homes/spielman/561/2012/lect03-12.pdf}.}
    \begin{lemma}
    Let $\mathcal{G}$ be a connected weighted graph (with non-negative edge weights), let ${\bf A}^{adj}$ be its adjacency matrix, and assume that some non-negative vector $\phi$ is an eigenvector of ${\bf A}^{adj}$. Then, $\phi$ is strictly positive.
    \end{lemma}
    
    Despite that in GRM the matrix ${\bf A}$ is not an adjacency matrix, its sparsity (the zero/non-zeros) pattern reveals the information of its corresponding adjacency matrix. As an approximation, the lemma entails one of the properties of the implied beta -- the majority elements of implied beta would be strictly positive. This is consistent with the empirical observation that the movement of an individual stock return follows the movement of the market.
    
    \item \textit{Elements of an implied beta fluctuate around 1}. This can be explained by the fact that an implied beta is the first eigenvector of $\bomg^{-1}$, which can be approximated by the first eigenvector of the inverse of the Laplacian matrix\footnote{The Laplacian matrix is defined\cite{Biggs1993} as ${\bf L} = {\bf D}^{deg} - {\bf A}^{adj}$ where ${\bf A}^{adj}$ is the adjacency matrix of the graph corresponding to $\bomg$ and ${\bf D}^{deg}$ is the degree matrix.} ${\bf L}$ of $\bomg$. The first eigenvector of the inverse of ${\bf L}$ is the vector of all ones $(1, 1, \cdots, 1)^{\top}$\cite{Biggs1993}. Thus, the implied beta can be approximated by the vector of all ones, and elements of implied beta should fluctuate around 1.
\end{enumerate}

\subsection{Graph visualization for asset returns}

We plot the graph for the component stocks of the S\&P 500 Index using the dataset from the expansion period. In the graph, each node represents a stock and each edge represents the element in the partial correlation matrix $\hat{\bf P}$ that is computed by $\req{eq:partial correlation matrix}$ where $\hat\bomg$ is estimated via either Glasso or Concord. Also, for the sake of comparison, we construct the graph for the matrix computed from a hard-thresholding PCA procedure. In this procedure:

\begin{enumerate}
    \item Compute the sample covariance matrix ${\bf S}$ of ${\bf Y}_I$ and its $k$ largest eigenvalues $\lambda_1,\cdots, \lambda_k$ and the corresponding eigenvectors $\beta_1, \cdots, \beta_k$. 
    
    \item Denote $\hat{\bf B} = (\beta_1, \cdots, \beta_k)$ and $\hat{\blam} = \diag(\lambda_1, \cdots, \lambda_k)$, compute $\hat\bdel = \diag({\bf S} - \hat{\bf B}\hat{\blam}\hat{\bf B}^{\top})$, then the PCA-based estimation for the precision matrix is $\hat\bomg = (\hat{\bf B}\hat{\blam}\hat{\bf B}^{\top} + \hat\bdel)^{-1}$.
    
    \item Based on $\req{eq:partial correlation matrix}$, compute the PCA-based partial correlation matrix $\hat{\bf P} = ((\hat\varrho_{ij}))$, take a threshold level $0<\gamma< 1$ and apply the hard-thresholding operation on $\hat{\bf P}$, i.e. 
    \begin{equation}
        \tilde{\bf P} = ((\tilde \varrho_{ij}))\quad\text{where}\quad
        \tilde \varrho_{ij} = 
        \begin{cases}
        \hat\varrho_{ij},&\text{ if $|\hat\varrho_{ij}| > \gamma$}\\
        0, &\text{ if $|\hat\varrho_{ij}|\leq\gamma$}
        \end{cases}
    \end{equation}
\end{enumerate}

In our empirical study, $k = 3$ and $\gamma$ is chose via the bisection algorithm such that $\tilde{\bf P}$ has the same number of non-zero elements as that of $\hat{\bf P}$ computed from Concord method. The graphs are shown in Figure \ref{fig:sp500}.

\begin{figure}[H]
\centering
  \includegraphics[scale = 0.42]{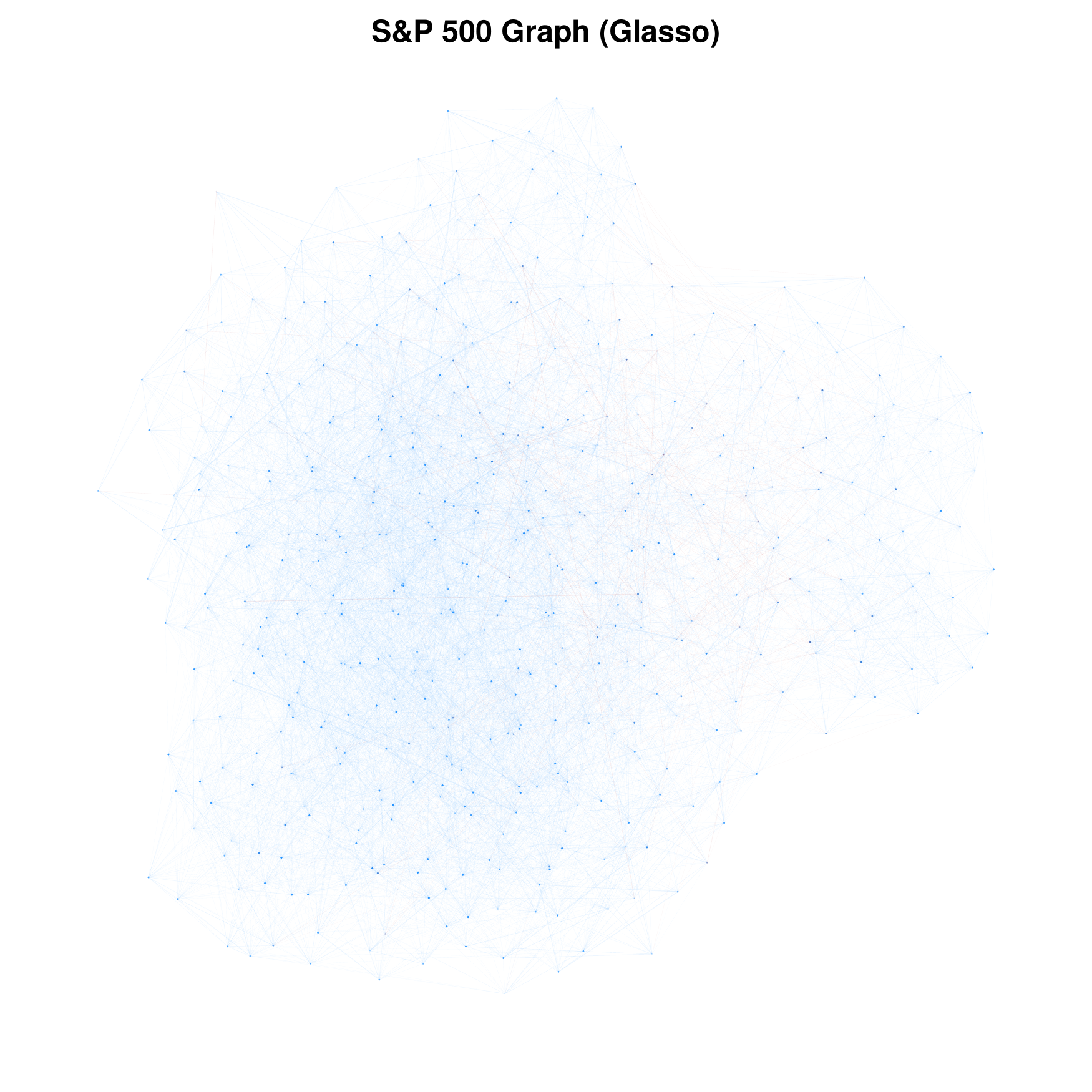}
  \includegraphics[scale = 0.42]{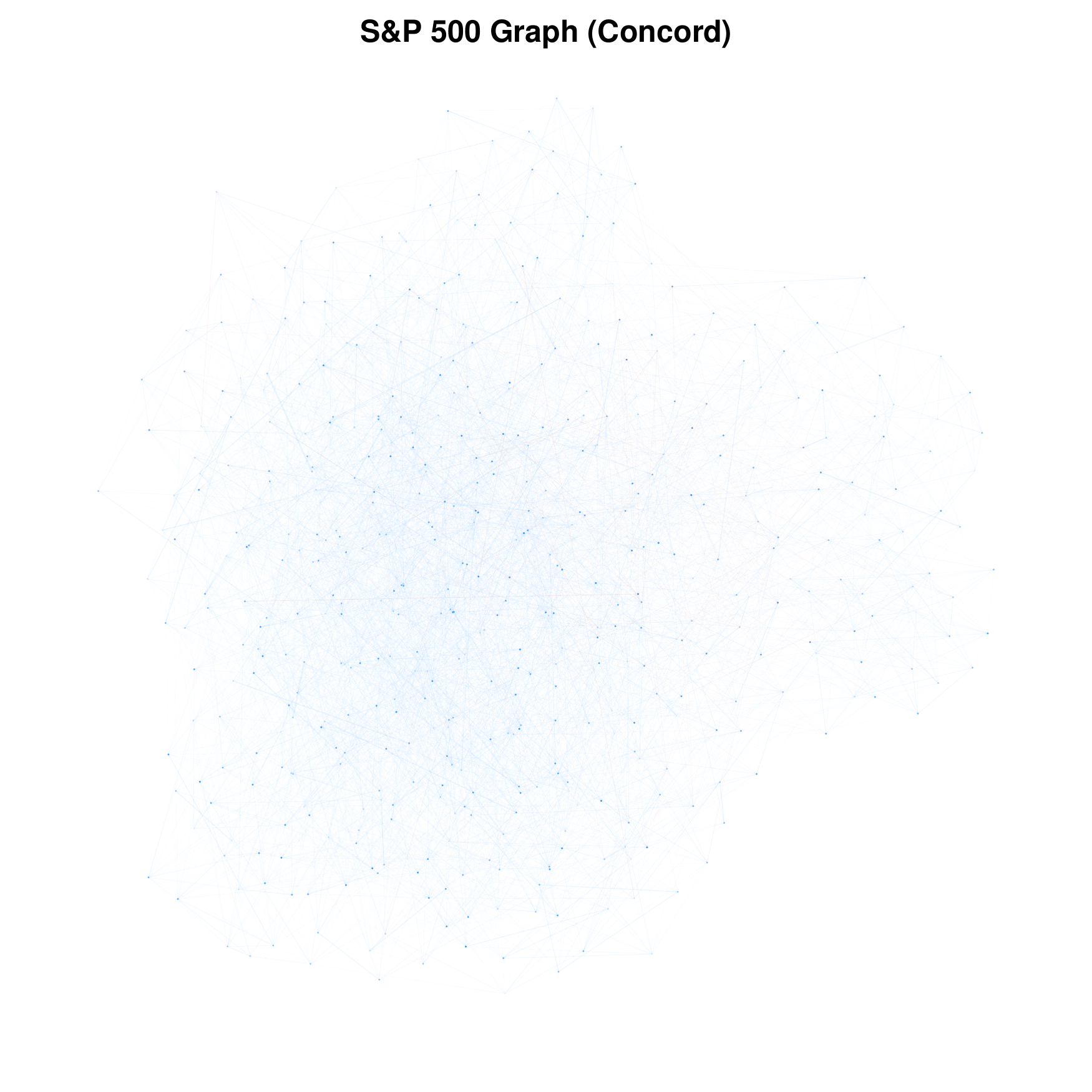}
\includegraphics[scale = 0.42]{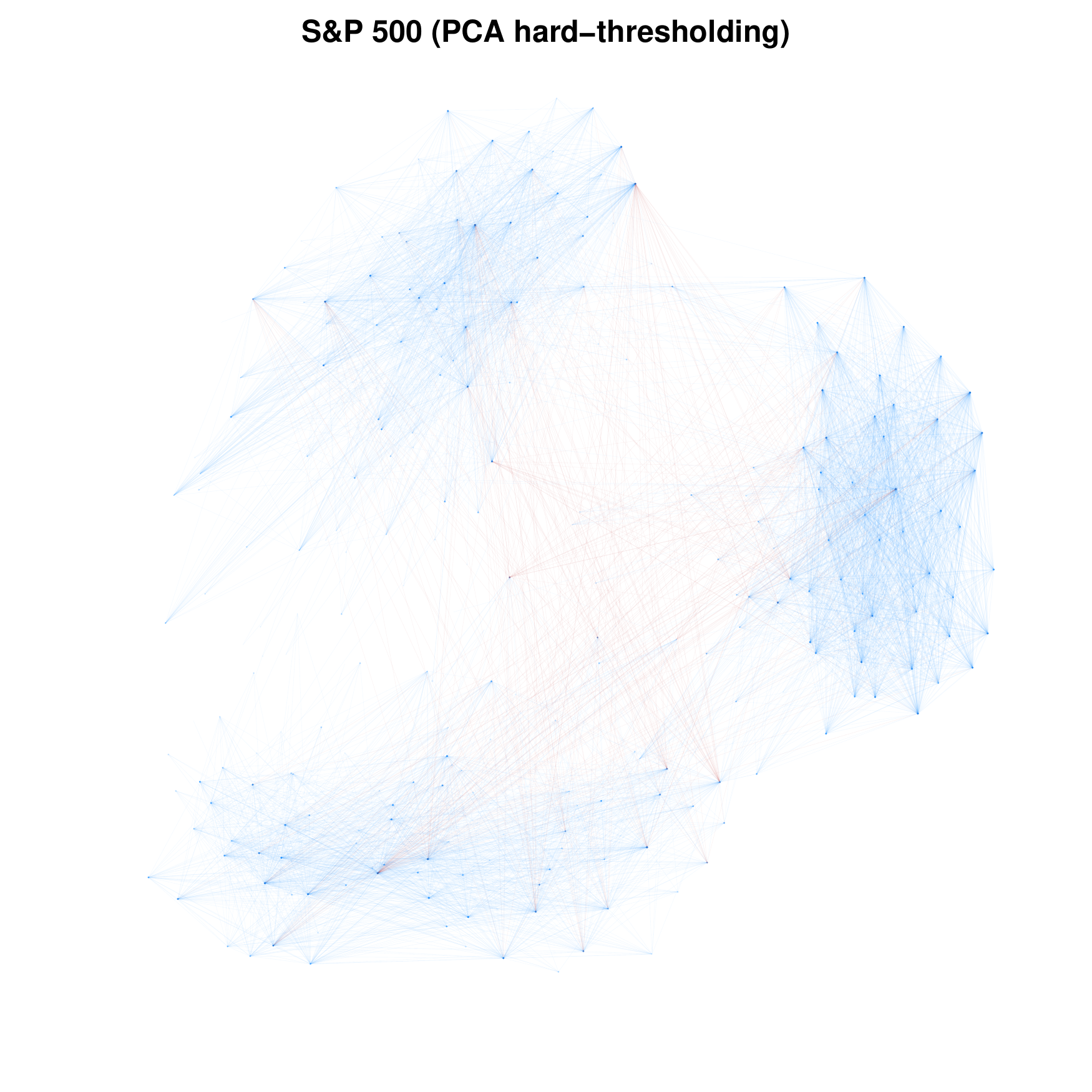}
\caption{The graphs of S\&P 500 Component Stocks built by the estimated partial correlation matrix $\hat{\bf P}$ coming from Glasso (top-left), Concord (top-right), and PCA hard-thresholding(bottom). These graphs are built using the S\&P 500 in-sample data of the economic expansion period, as shown in Table \ref{tab:time}. All graphs are presented with the same layout. A blue edge represents a positive partial correlation while a red edge represents a negative partial correlation, with the thickness of an edge proportional to the correlation value. }
\label{fig:sp500}
\end{figure}

\subsection{Sector/community analysis in graphs}
In this sub-section, we are interested in the sector/community analysis from the graph, since graphs in Figure \ref{fig:sp500} seem to contain some clusters. Within these clusters, more blue edges exist while between the clusters more red edges appear. One may expect that if this pattern is related to the different industrial sectors or data-driven communities. Therefore, exploring the networks in terms of sectors or data-driven communities detection is of great interest to us. We download the sector information for each stock of S\&P 500 Index\footnote{ The sector categories can be found in the column of Global Industry Classification Standard (GICS) sectors via \url{https://en.wikipedia.org/wiki/List_of_S\%26P_500_companies}. For all S\&P 500 component stocks, there are eleven sectors in total: industrials, health care, information technology, communication services, consumer discretionary, utilities, financials, materials, real estate, consumer staples, and energy.} There are 11 sectors in total. Suppose $S_i$ is the set of the index of all assets belonging to the $i$-th sector, $|S|$ denotes the cardinality of a set $S$. Define:

\begin{equation}
    \begin{split}
        C_{ij}^+ &\triangleq \Big|\{(h,l):\hat\varrho_{hl} > 0 \text{ for $h\in S_i$ and $l\in S_j$ and $h\neq l$}\}\Big|\\
        C_{ij}^- &\triangleq \Big|\{(h,l):\hat\varrho_{hl} < 0 \text{ for $h\in S_i$ and $l\in S_j$ and $h\neq l$}\}\Big|
    \end{split}
\end{equation}

\noindent To visualize the relative counts of blue edges and red edges, we design the ratio matrix $\bphi = ((\phi_{ij}))$ where $\phi_{ij} = \log\Big( \frac{C_{ij}^++1}{C_{ij}^-+1}\Big)$. The scaled matrix $\tilde\bphi = \frac{1}{\max_{i,j}(\phi_{ij})}\bphi$ is visualized in Figure \ref{fig:sp500mat}. The darker the blue, the more blue edges exist relative to red edges. The color pattern of both graphs indicates that assets within sectors have relatively more blue edges than red edges thus are more likely to positively partially correlated than assets between sectors. Besides, the empty color is the dominated color for all sectors, implying that the entire market has a co-movement to some extent. This result is consistent with the empirical phenomenon that the majority of market betas are positive \cite{Blume1971}, that is, when the market goes up, all the returns go up, and when the market goes down, the opposite is true.

Since the sector information is externally provided, we are curious about if the data-driven detected communities have any correspondence with these sectors. One of the community detection methods, called communities in a graph via random walks and proposed by \cite{Pons2005}, tries to find densely connected subgraphs via detecting short random walks. We use the \texttt{cluster\_walktrap} function from the \textbf{R} package \texttt{igraph}\footnote{The package and its reference manual can be downloaded from \url{https://cran.r-project.org/web/packages/igraph/index.html}.} to conduct community detection over the sparsity patterns of the estimated precision matrix $\hat\bomg$. The communities are detected using random walks method which is proposed by \cite{Pons2005}. The detection is based on the zero/non-zero patterns of $\hat\bomg$ computed from either Glasso or Concord. A dendrogram is created and we cut the dendrogram such that all stocks are partitioned into 11 clusters. Table \ref{tab:sc-count} and \ref{tab:sc-count2} show the stock counts falling into each sector and community using the Glasso or the Concord estimator. The results show that there are some clusters that are dominated by certain sectors. We are interested in comparing the result of blue/red edge counts from the community detection with that from the sector categorization. We assume that $S_i$ in the context of community detection represents the set of the index of all assets belonging to the $i$-th cluster. We compute $C_{ij}^+, C_{ij}^-$ for all $i$ and $j$, and obtain $\tilde\bphi$. The visualization of $\tilde\bphi$ is shown in Figure \ref{fig:sp500mat2}.

\begin{figure}[H]
\centering
\begin{subfigure}{0.49\textwidth}
  \includegraphics[width=\textwidth]{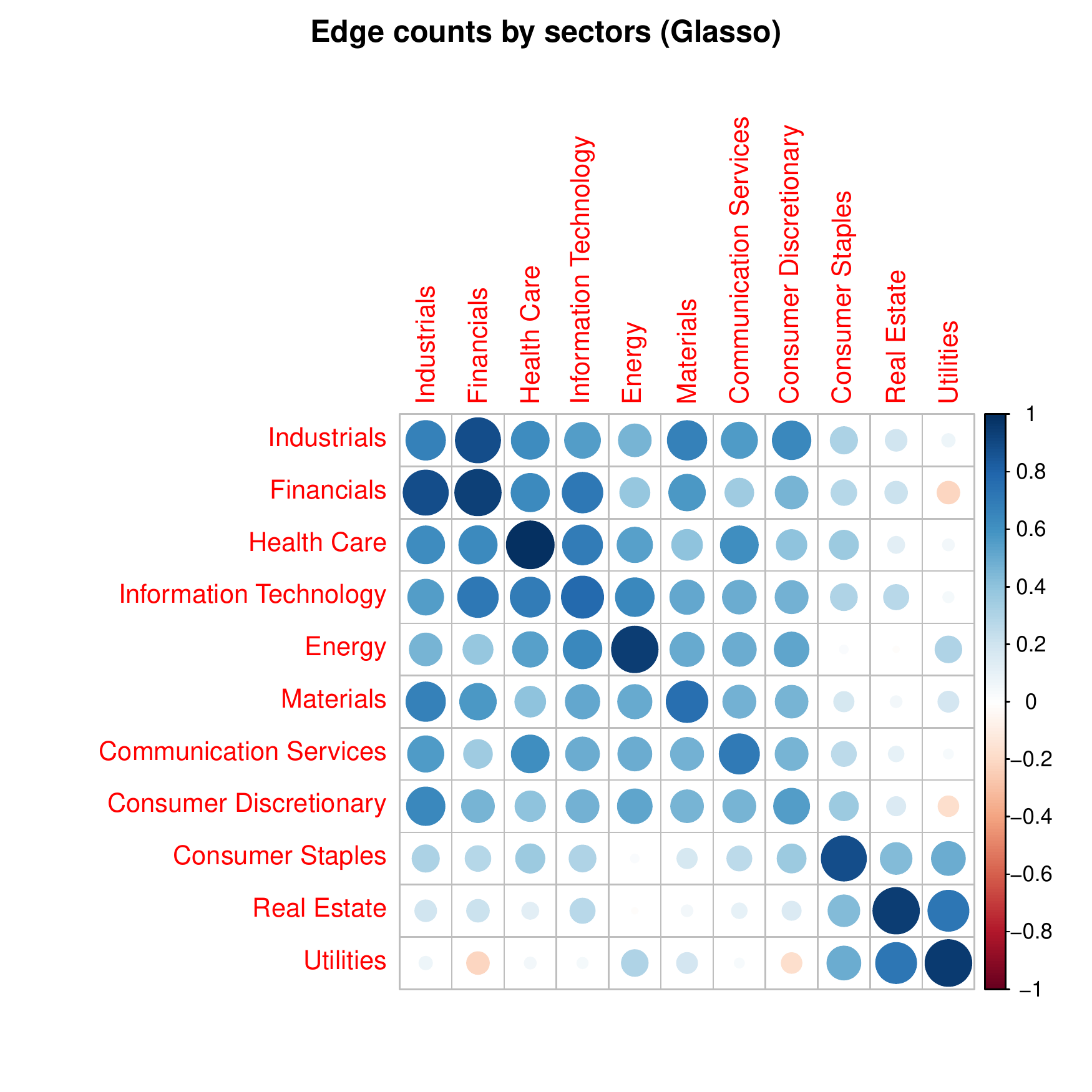}
  \end{subfigure}
  \begin{subfigure}{0.49\textwidth}
  \includegraphics[width=\textwidth]{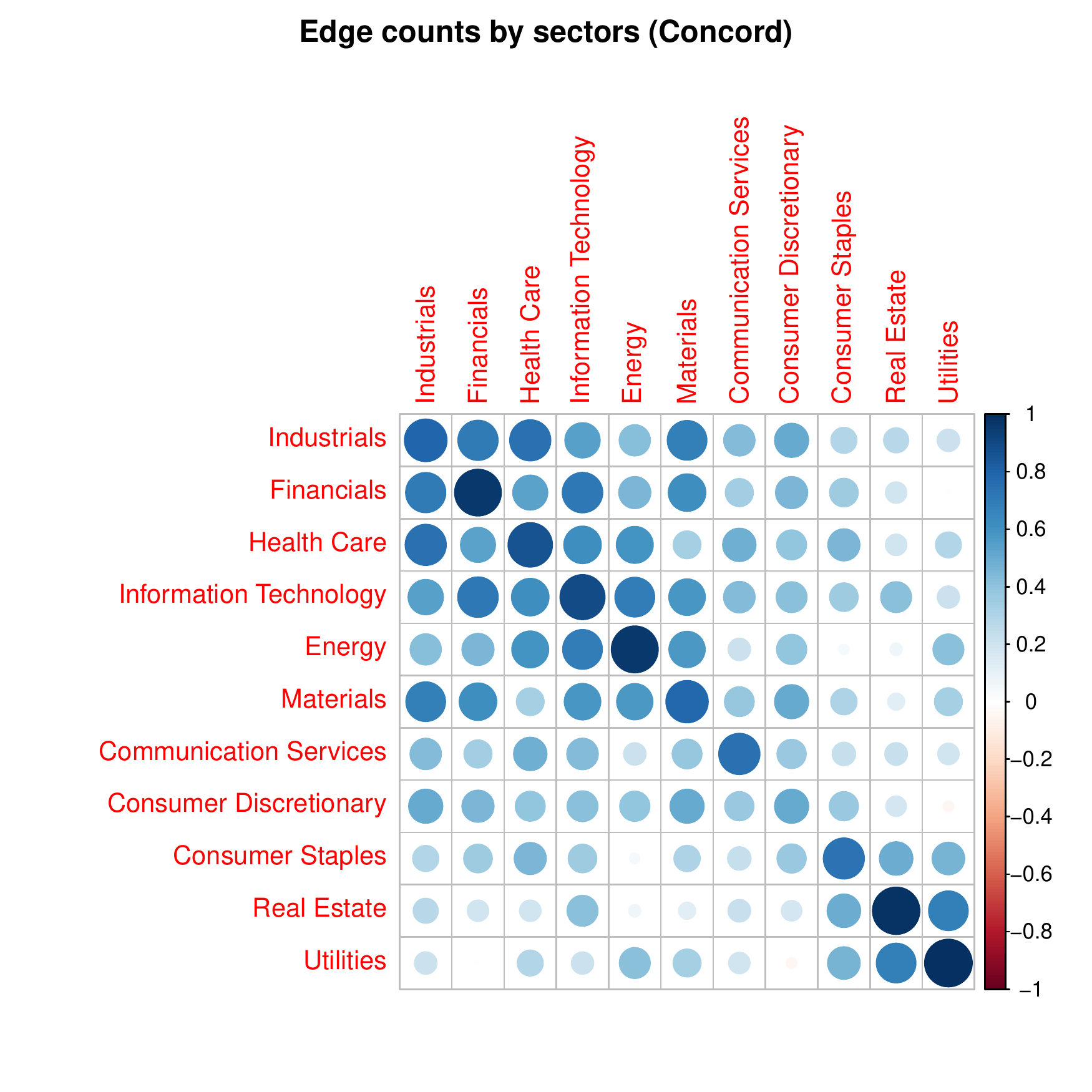}
  \end{subfigure}
\caption{The visualization of dependencies among different sectors, encoded by $\tilde\bphi$, that are measured by the relative blue edge counts to red edge counts within sectors and between sectors. The blue/red edge counts are computed from the sparse precision matrix estimated from graphical lasso (above) and Concord (below). The sparse precision matrices are estimated using the in-sample data of the S\&P 500 component stock historical data from the economic expansion period. The darker the blue circle in the plots, the more blue edges exist relative to red edges, thus the more positive partial correlations exists relative to the negative partial correlations.}
\label{fig:sp500mat}
\end{figure}

\begin{figure}[H]
\centering
\begin{subfigure}{0.49\textwidth}
  \includegraphics[width=\textwidth]{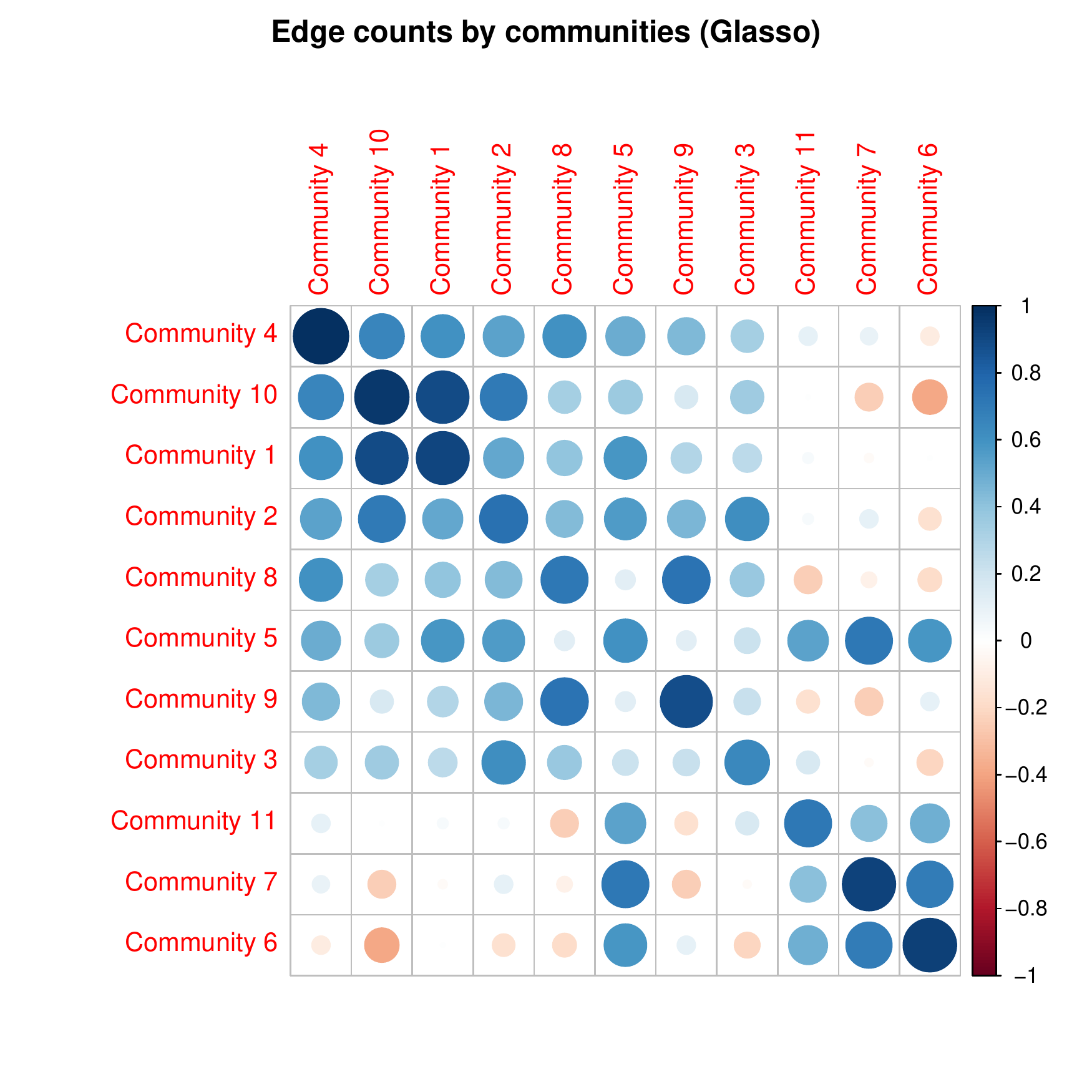}
  \end{subfigure}
  \begin{subfigure}{0.49\textwidth}
  \includegraphics[width=\textwidth]{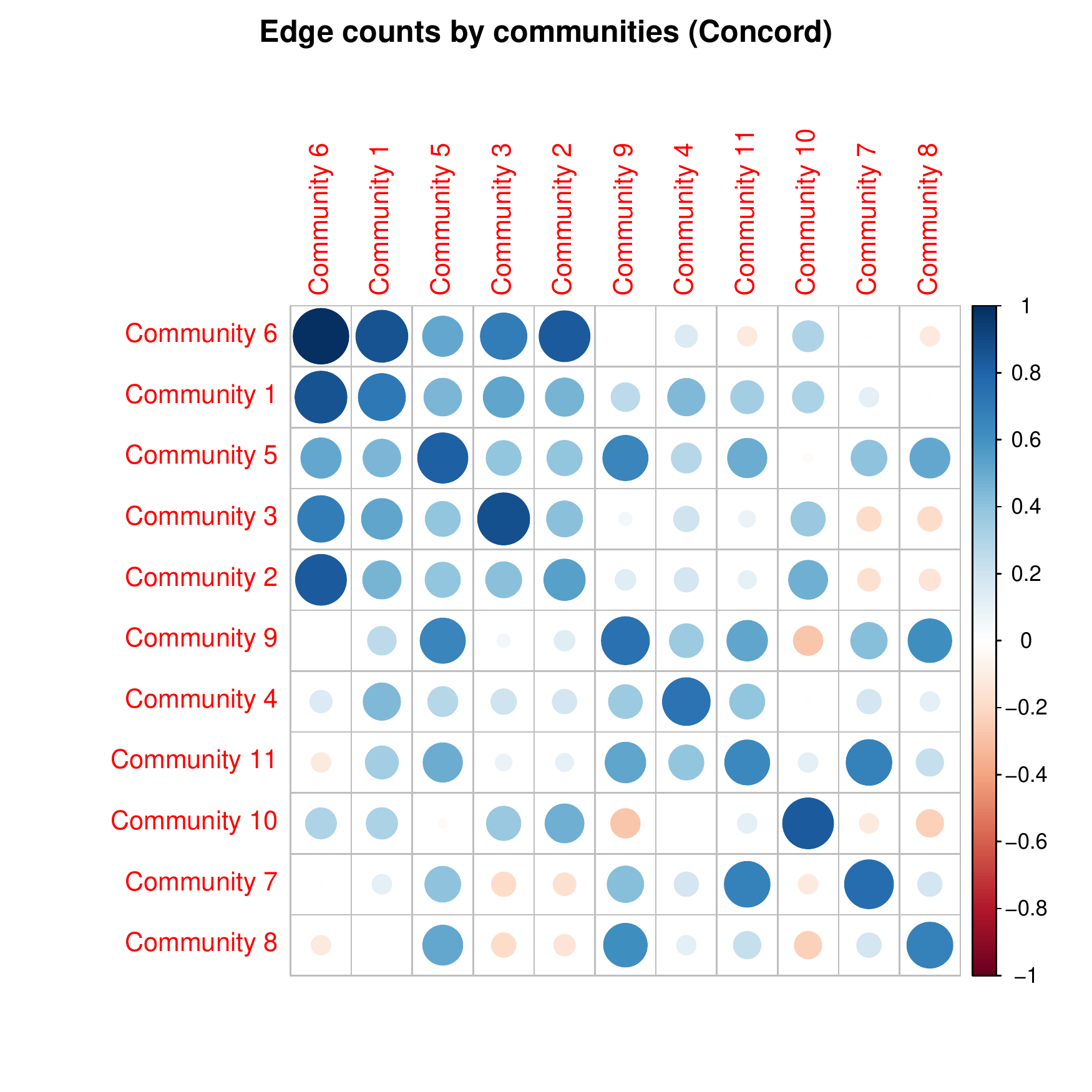}
  \end{subfigure}
\caption{The visualization of dependencies among different communities that are measured by the relative blue edge counts to red edge counts within communities and between communities. The blue/red edge counts are computed from the sparse precision matrix estimated from graphical lasso (above) and Concord (below). The sparse precision matrices are estimated using the in-sample data of the S\&P 500 component stock historical data from the economic expansion period. The darker the blue circle in the plots, the more blue edges exist relative to red edges, thus the more positive partial correlations exists relative to the negative partial correlations.}
\label{fig:sp500mat2}
\end{figure}

\begin{table}[ht!]
    \centering
    \begin{tabular}{c|ccccc ccccc c}
    \hline\hline
    \multicolumn{12}{c}{Stock counts in each sector and community (Glasso)}\\\hline
Community Label       &  1&  2&  3&  4&  5&  6&  7& 8&  9& 10& 11\\\hline
Health Care           & 20& 26&  5&  6&  2&  0&  1& 1&  0&  0&  0\\
Industrials           & 29& 34&  1&  0&  3&  0&  0& 0&  0&  0&  0\\
Consumer Discretionary&  8& 25& 19&  4&  1&  0&  0& 6&  0&  0&  0\\
Information Technology& 18&  9&  0& 32&  6&  0&  0& 2&  0&  0&  0\\
Consumer Staples      &  1&  2&  3&  0& 13&  0&  1& 0&  0&  0& 13\\
Utilities             &  0&  0&  1&  0&  1& 25&  1& 0&  0&  0&  0\\
Financials            & 32&  3&  0&  2&  2&  0&  0& 0&  0& 29&  0\\
Real Estate           &  0&  1&  0&  0&  7&  0& 24& 0&  0&  0&  0\\
Materials             &  2&  9&  2&  2&  2&  0&  0& 8&  0&  0&  0\\
Energy                &  0&  1&  0&  0&  0&  0&  0& 6& 22&  0&  0\\
Communication Services&  1&  5&  9&  6&  4&  0&  0& 0&  0&  0&  0\\\hline
{\bf Total}           &111& 115&  40&  52&  41&  25&  27&  23&  22&  29&  13\\\hline\hline
\end{tabular}
\caption{The results of community detection for S\&P 500 component stocks using the in-sample data of the S\&P 500 historical data from the economic expansion period. The community detection is conducted over the sparse precision matrix estimated from graphical lasso. The numbers represent the stock counts in each sector and community. The results show that there are some communities that are dominated by stocks coming from one certain sector.}
\label{tab:sc-count}
\end{table}

\begin{table}
\centering
    \begin{tabular}{c|ccccc ccccc c}
    \hline\hline
    \multicolumn{12}{c}{Stock counts in each sector and community (Concord)}\\\hline
Community Label       &  1&  2&  3&  4& 5&  6&  7&  8&  9& 10& 11\\\hline
Health Care           & 30& 19&  9&  0& 2&  1&  0&  0&  0&  0&  0\\
Industrials           & 30& 35&  1&  0& 1&  0&  0&  0&  0&  0&  0\\
Consumer Discretionary& 11& 41&  4&  0& 5&  1&  0&  0&  0&  1&  0\\
Information Technology& 15& 14& 32&  0& 5&  1&  0&  0&  0&  0&  0\\
Consumer Staples      &  9&  5&  0& 15& 3&  0&  0&  0&  1&  0&  0\\
Utilities             &  0&  1&  0&  0& 1&  0& 13&  0&  1&  0& 12\\
Financials            & 10& 12&  7&  0& 3& 36&  0&  0&  0&  0&  0\\
Real Estate           &  1&  1&  0&  0& 5&  0&  0& 10& 15&  0&  0\\
Materials             &  6& 13&  3&  0& 3&  0&  0&  0&  0&  0&  0\\
Energy                &  0& 10&  0&  0& 1&  0&  0&  0&  0& 18&  0\\
Communication Services&  1& 14&  6&  1& 3&  0&  0&  0&  0&  0&  0\\\hline
{\bf Total}           & 113& 165&  62&  16&  32&  39&  13&  10&  17&  19&  12 \\\hline
    \end{tabular}
    \caption{A continued table for Table \ref{tab:sc-count}. The results of community detection for S\&P 500 component stocks using the in-sample data of the S\&P 500 historical data from the economic expansion period. The community detection is conducted over the sparse precision matrix estimated from Concord method.}
\label{tab:sc-count2}
\end{table}

\subsection{Communities detection from the graphs}

We choose the return data of S\&P 500 component stocks during the expansion period, obtain the estimated the precision matrix $\hat\bomg$ via Glasso, and conduct the community detection based on the sparsity patterns of $\hat\bomg$. Figure \ref{fig:cmm3-4}, \ref{fig:cmm6-7}, \ref{fig:cmm8-9}, and \ref{fig:cmm10-11} show Community 3, 4, and 6 - 11.

\begin{figure}[H]
\centering
\begin{subfigure}{0.49\textwidth}
  \includegraphics[width=\textwidth]{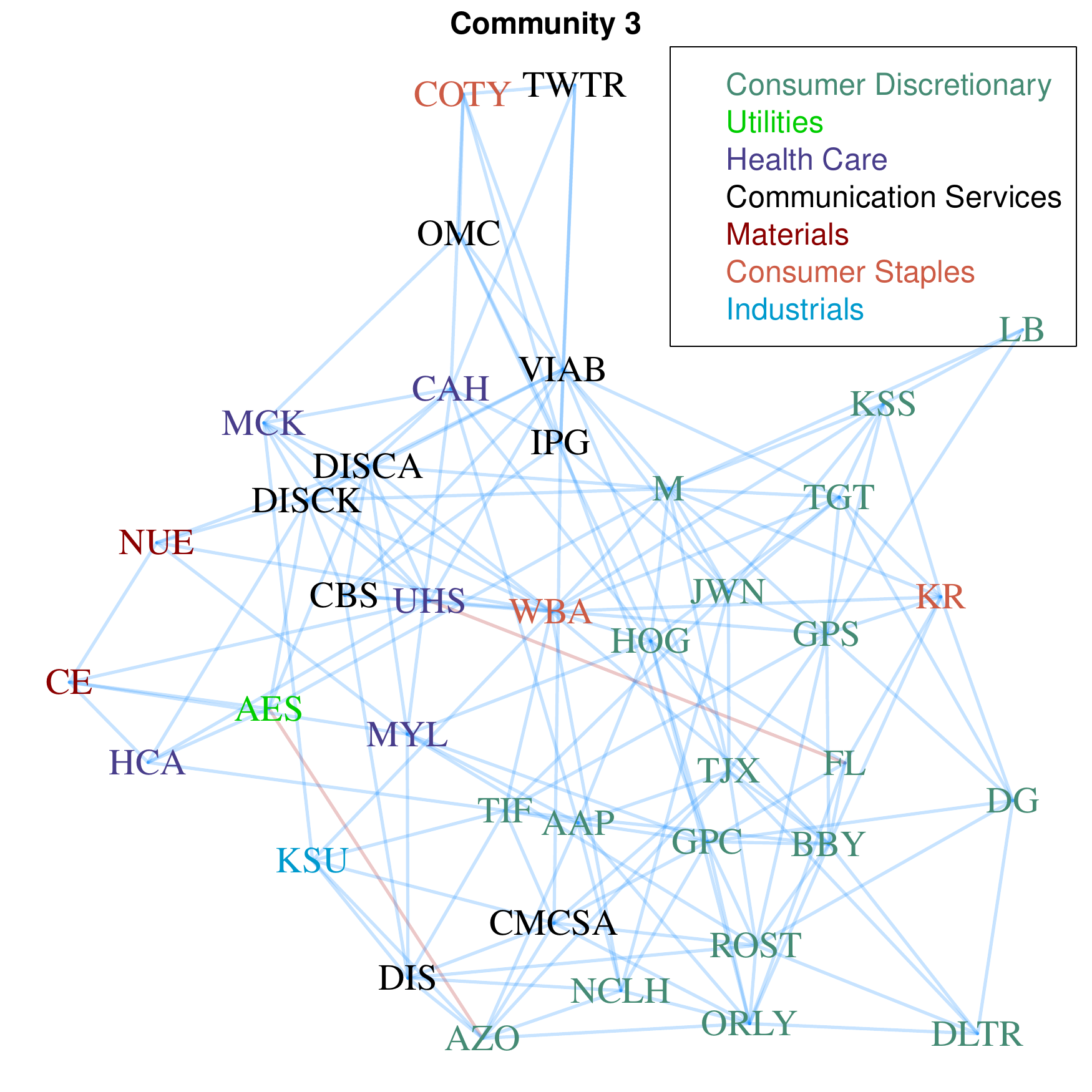}
  \end{subfigure}
  \begin{subfigure}{0.49\textwidth}
  \includegraphics[width=\textwidth]{full-cmm-glasso-c4.pdf}
  \end{subfigure}
\caption{The recovered graphs of two communities (Community 3 \& 4) from the community detection based on the sparse precision matrix estimated from graphical lasso. The sparse precision matrix is estimated using the in-sample data of the S\&P 500 component stock historical data from the economic expansion period. The above community is dominated by stocks from consumer discretionary sector while the below community is dominated by stocks from information technology sector.}
\label{fig:cmm3-4}
\end{figure}

\begin{figure}[H]
\centering
\begin{subfigure}{0.49\textwidth}
\includegraphics[width = \textwidth]{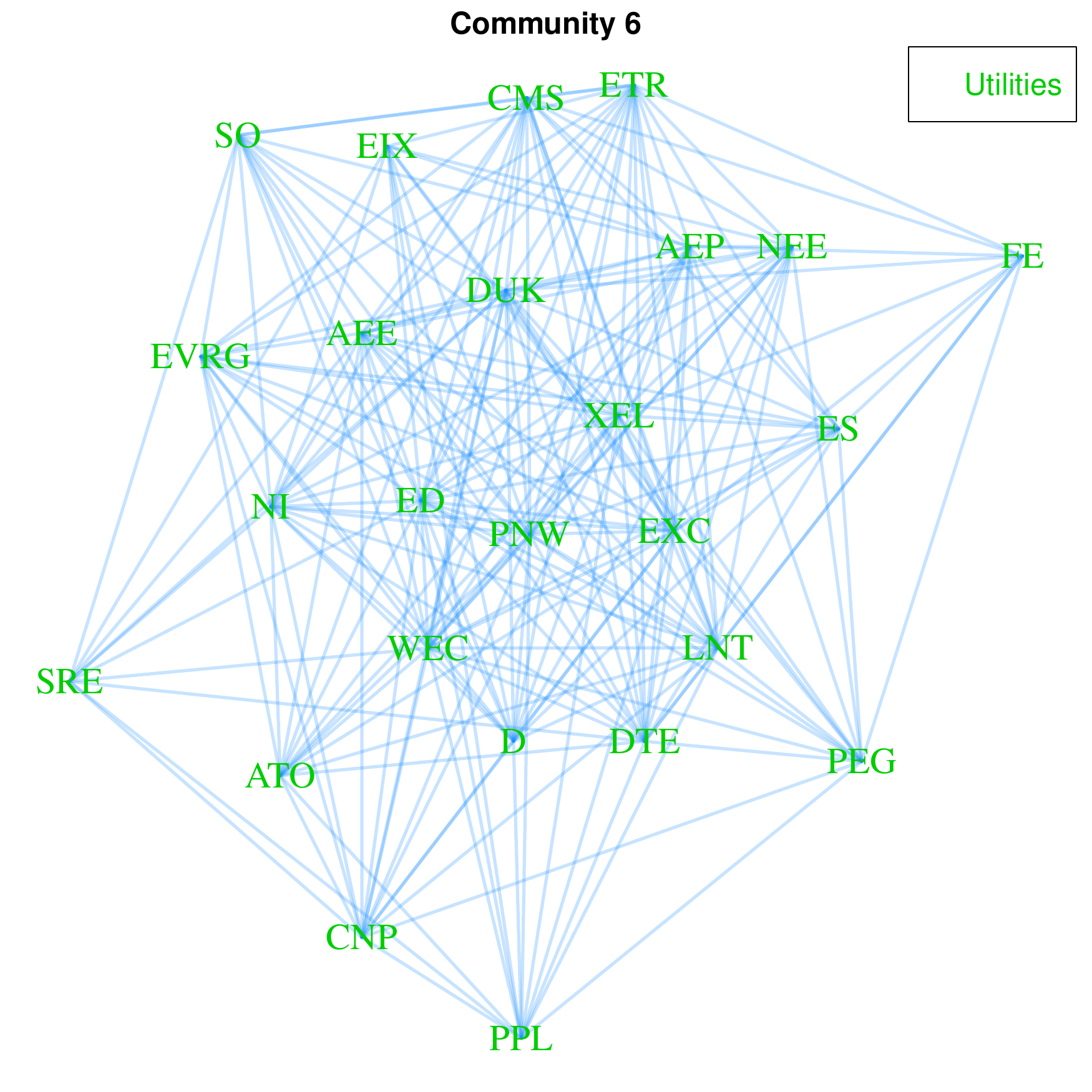}
\end{subfigure}
\begin{subfigure}{0.49\textwidth}
  \includegraphics[width = \textwidth]{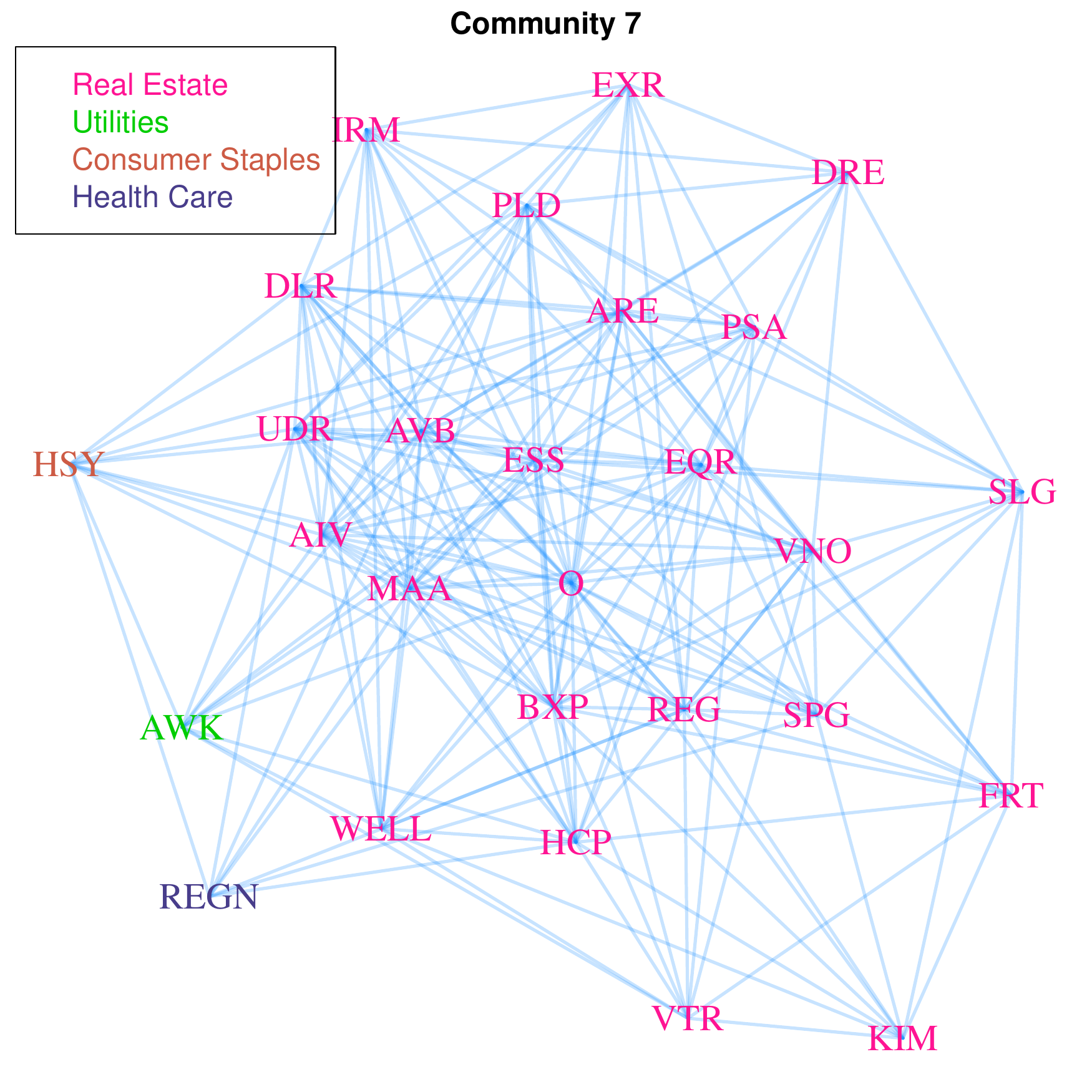}
  \end{subfigure}
\caption{The recovered graphs of two communities (Community 6 \& 7) from the community detection based on the sparse precision matrix estimated from graphical lasso. The sparse precision matrix is estimated using the in-sample data of the S\&P 500 component stock historical data from the economic expansion period. The above community is dominated by stocks from utilities sector while the below community is dominated by stocks from real estate sector.}
\label{fig:cmm6-7}
\end{figure}

\begin{figure}[H]
\centering
\begin{subfigure}{0.49\textwidth}
\includegraphics[width = \textwidth]{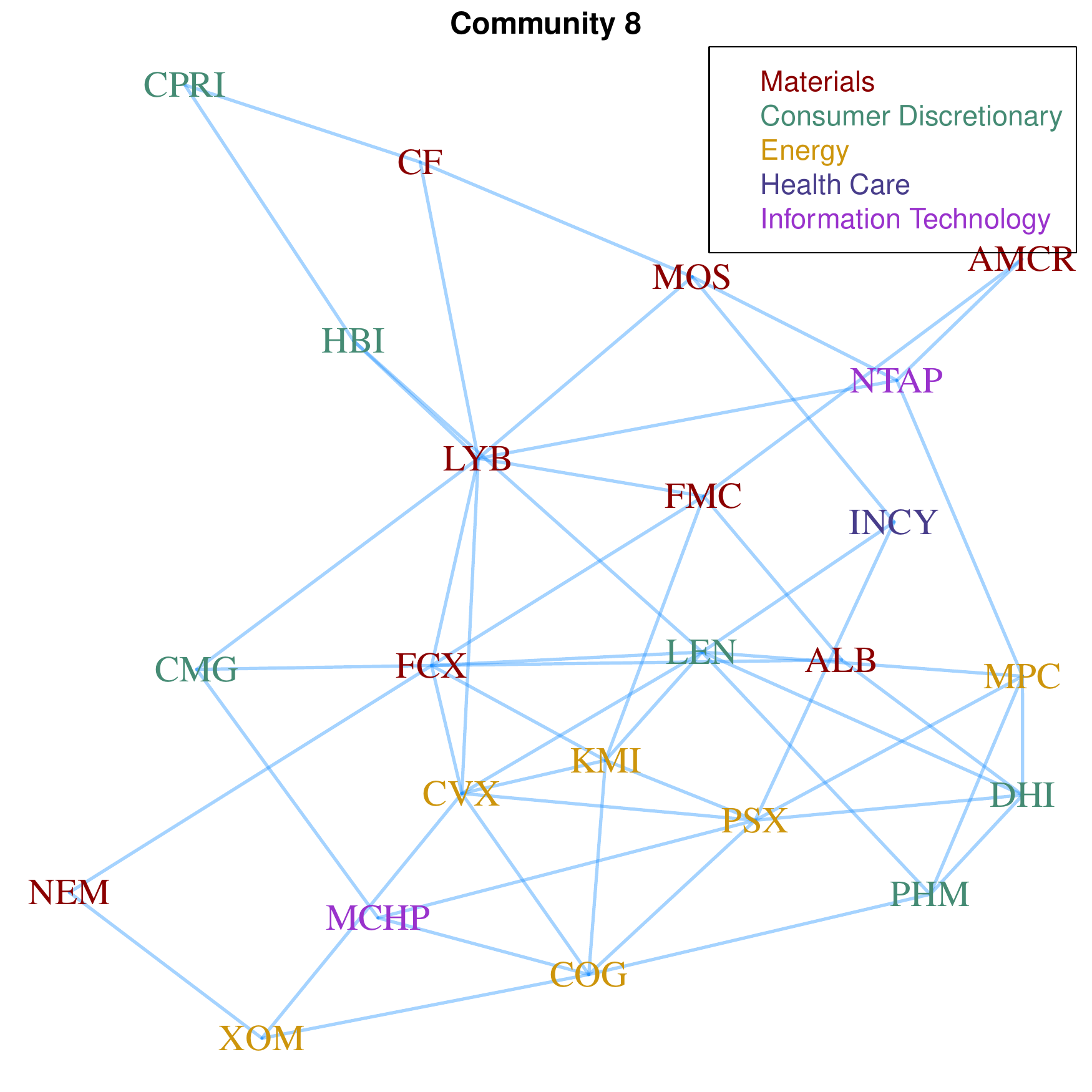}
\end{subfigure}
\begin{subfigure}{0.49\textwidth}
  \includegraphics[width = \textwidth]{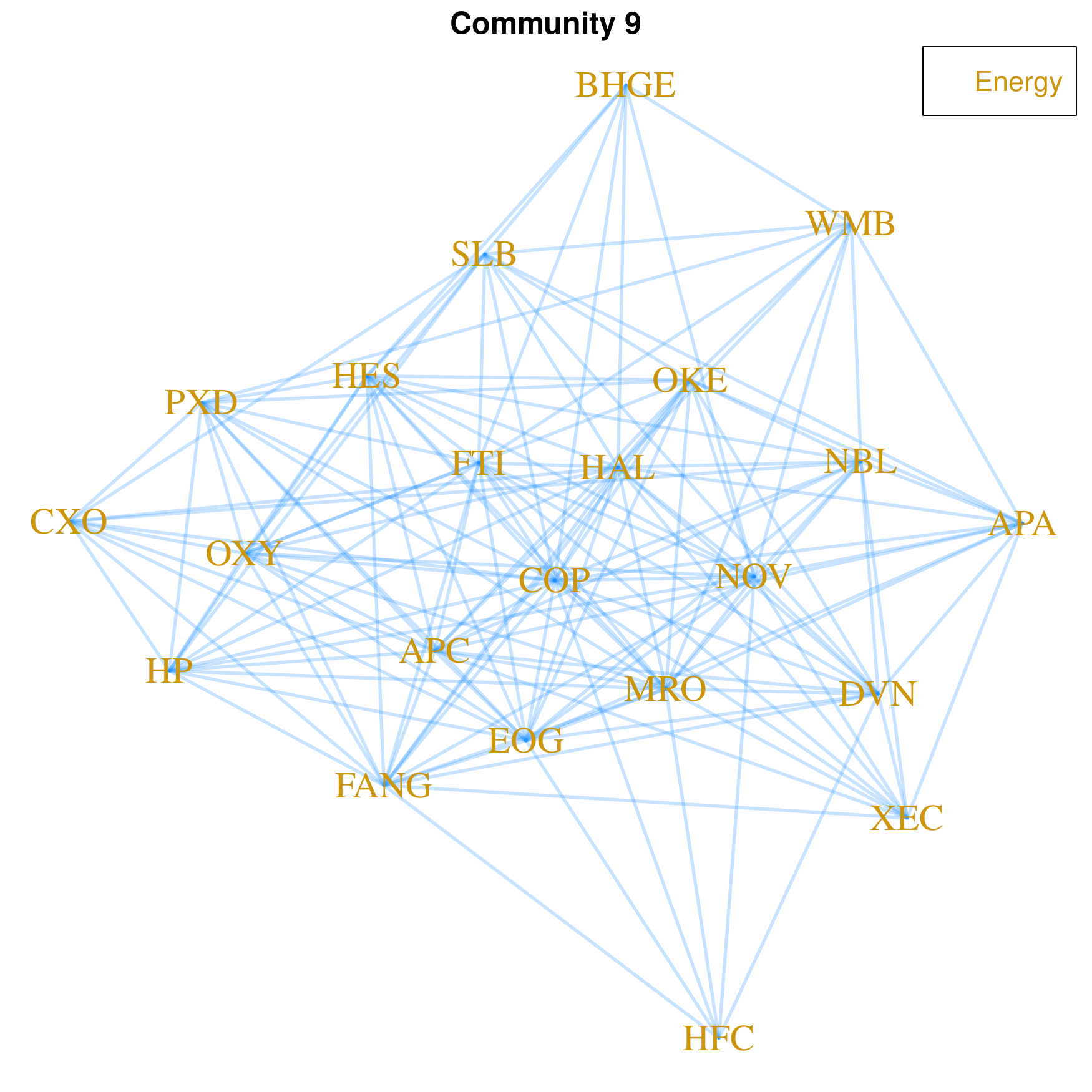}
  \end{subfigure}
\caption{The recovered graphs of two communities (Community 8 \& 9) from the community detection based on the sparse precision matrix estimated from graphical lasso. The sparse precision matrix is estimated using the in-sample data of the S\&P 500 component stock historical data from the economic expansion period. The above community is not dominated by stocks from a certain sector while the below community is dominated by stocks from energy sector.}
\label{fig:cmm8-9}
\end{figure}

\begin{figure}[H]
\centering
\begin{subfigure}{0.48\textwidth}
 \includegraphics[width = \textwidth]{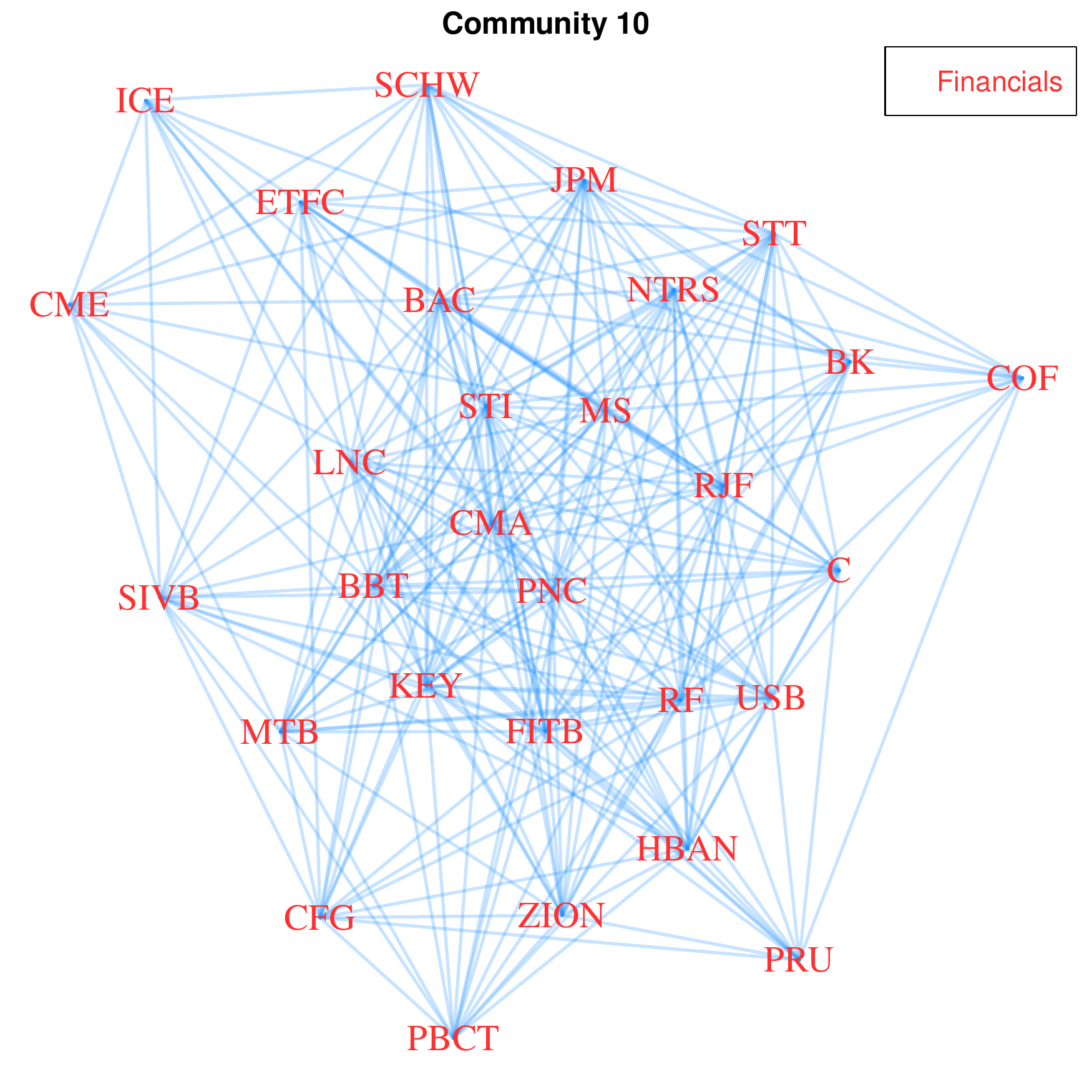}
 \end{subfigure}
 \begin{subfigure}{0.48\textwidth}
  \includegraphics[width = \textwidth]{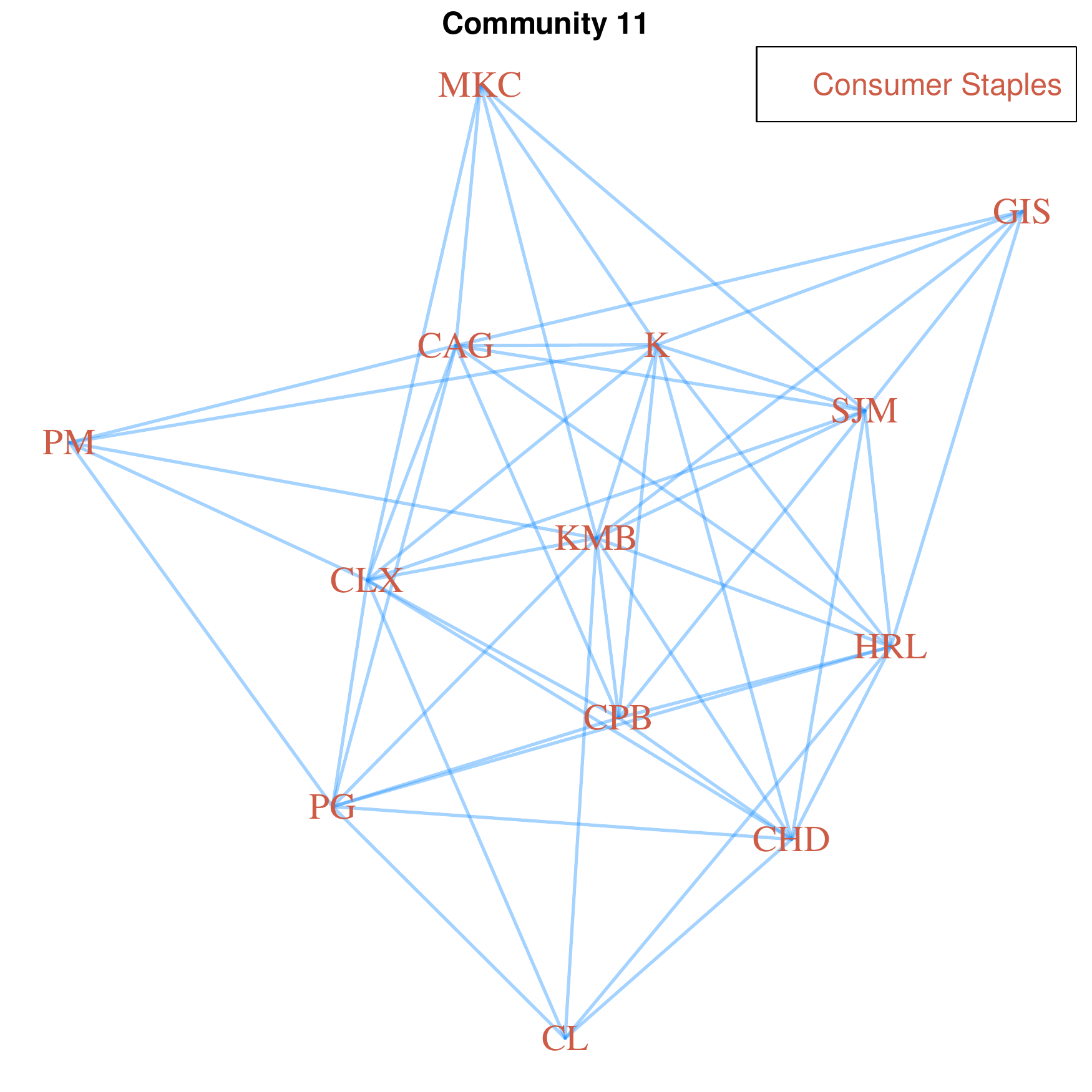}
  \end{subfigure}
\caption{The recovered graphs of two communities (Community 10 \& 11) from the community detection based on the sparse precision matrix estimated from graphical lasso. The sparse precision matrix is estimated using the in-sample data of the S\&P 500 component stock historical data from the economic expansion period. The above community is dominated by stocks from financial sector while the below community is dominated by stocks from consumer staples sector.}
\label{fig:cmm10-11}
\end{figure}

\bibliographystyle{unsrt}  
\bibliography{dissertation.bib}  


\end{document}